\documentclass[10pt,twocolumn,twoside]{IEEEtran} 
\usepackage{amssymb}

\let\labelindent\relax \usepackage{enumitem}  
\usepackage{mathtools,amsthm}

\usepackage{graphics} 
\usepackage{epsfig}
\usepackage[linesnumbered,boxed,ruled,commentsnumbered]{algorithm2e}
\usepackage{amsmath,amssymb,amsfonts}
\usepackage{mathrsfs}
\usepackage{cite}
\usepackage{siunitx}
\usepackage{algorithmic}
\usepackage{graphicx}
\usepackage{textcomp}
\usepackage{cuted}
\usepackage{lipsum}
\usepackage{varioref}
\usepackage{import}
\usepackage{framed,xcolor}
\usepackage{color}
\usepackage{comment}
\usepackage{multirow}
\usepackage{epstopdf}   
\usepackage{psfrag}
\usepackage{breqn}
\usepackage{float}
\usepackage{mathtools}
\usepackage{booktabs}
\usepackage{soul}
\usepackage{amsbsy}
\usepackage[bottom]{footmisc}
\usepackage{lineno}
\usepackage{cleveref}
\usepackage{microtype}
\usepackage{comment}
\usepackage{epstopdf}

\newcommand{\vect}[1]{\ensuremath{\boldsymbol{\mathrm{#1}}}}
\newtheorem{theorem}{Theorem}
\newtheorem{Proposition}{Proposition}
\newtheorem{Corollary}{Corollary}
\newtheorem{Lemma}{Lemma}

\newtheorem{Remark}{Remark}

\newtheorem{Assumption}{Assumption}

\DeclareMathOperator*{\argmin}{arg\,min}

\IEEEoverridecommandlockouts                              

\title{\LARGE \bf
Distributed Model Predictive Control Design for Multi-agent Systems via Bayesian Optimization}

\author{Hossein Nejatbakhsh Esfahani, Kai Liu, Javad Mohammadpour Velni
\thanks{*This work was supported by the US National Science Foundation under award \#2302219.}
\thanks{H. N. Esfahani and J. M. Velni are with the Department of Mechanical Engineering, Clemson University, Clemson, SC, USA. K. Liu is with the School of Computing, Computer Science Division, Clemson University, Clemson, SC, USA.
		{\tt\small \{hnejatb, kail, javadm\}@clemson.edu%
		}.
}}

\begin{document}

\maketitle

\begin{abstract}

This paper introduces a new approach that leverages Multi-agent Bayesian Optimization (MABO) to design Distributed Model Predictive Control (DMPC) schemes for multi-agent systems. The primary objective is to learn optimal DMPC schemes even when local model predictive controllers rely on imperfect local models. The proposed method invokes a dual decomposition-based distributed optimization framework, incorporating an Alternating Direction Method of Multipliers (ADMM)-based MABO algorithm to enable coordinated learning of parameterized DMPC schemes. This enhances the closed-loop performance of local controllers, despite discrepancies between their models and the actual multi-agent system dynamics. In addition to the newly proposed algorithms, this work also provides rigorous proofs establishing the optimality and convergence of the underlying learning method. Finally, numerical examples are given to demonstrate the efficacy of the proposed MABO-based learning approach.

\end{abstract}


\section{Introduction}

Networked systems have grown significantly in scale and complexity, leading to increasingly demanding requirements for model-based control accuracy and real-time performance. For such large-scale network systems, complex optimization algorithms are commonly employed to perform desired tasks while minimizing a priori defined objective functions and satisfying operational constraints \cite{YANG2019278}. Optimization-based control of networked systems then remains a challenging problem when a complex model of multi-agent system is exploited within the optimization framework. Although simplified models can reduce the complexity of the model-based optimization algorithms, the inaccuracies can significantly degrade the performance of multi-agent control system.  

A common approach to managing large-scale networked systems is to design local controllers that neglect interactions between the subsystems involved. However, this often results in an overall degradation in the system performance. Although centralized control could achieve better global performance, it is generally impractical due to limitations in communication, the complexity of nonlinear systems, and the large number of decision variables \cite{Bemporad2010}. To address these challenges, extensive research has focused on structured control systems, including decentralized and distributed architectures. However, the primary challenge in a decentralized architecture lies in the lack of communication between controllers, which can lead to degraded closed-loop performance and, in some instances, even result in instability. In contrast, a distributed control system allows communication between controllers, enabling better coordination and potentially improved performance \cite{6853439}.

In the context of distributed control systems, distributed MPC (DMPC) is a well-researched optimal control approach that can handle interconnected systems and multi-variable interactions. Dual decomposition and Alternating Direction Method of Multipliers (ADMM) are two efficient methods for solving DMPC problems, where a coupled constraint between agents can be formulated as a dual problem \cite{Farokhi2014,dual-DMPC}. However, inaccurate local MPC models adversely affect the performance of DMPC. It can also be challenging to choose parameters for the cost and constraint functions for each agent. Machine learning-based approaches have recently emerged to address the aforementioned issue by enabling data-driven model learning for MPC \cite{8909368,9867259}. However, closed-loop control performance is not directly related to the model fitting, so that the control objectives may not be satisfied even if the learned model can accurately capture the real plant. To address this issue, the methods proposed in \cite{ntnu2024rlcontrol,dataDeriven2019,10808167, 10542325,10644368} established tools to fuse Markov Decision Process (MDP) and MPC so that an MPC scheme can deliver the same optimal policy as MDP by modifying the terminal and stage cost functions of the MPC. Then, several Reinforcement Learning (RL)-based methods were developed to learn the corresponding MPC cost functions aimed at improving closed-loop performance. It was shown that parameterizing the cost function can offer an alternative approach to model learning, motivated by the inherent connection between model-based predictions and the MPC cost function.

The integration of Bayesian Optimization (BO) with MPC has recently introduced a powerful paradigm for addressing complex control problems in various domains. By enabling adaptive, data-driven, and efficient control strategies, this combination holds the potential for many applications ranging from robotics to energy systems and beyond \cite{letham2019bayesian,gao2020energy,maheshwari2022bayesian}. BO is a probabilistic optimization approach designed for black-box functions that are expensive to evaluate. The methodology revolves around the construction of a surrogate model, often a Gaussian Process (GP), to approximate the target function \cite{7352306}. Appropriate fusion of BO with MPC leverages the strengths of both; while BO has the ability to optimize black-box functions, MPC offers predictive control capabilities. MPC performance hinges on the proper selection of parameters such as prediction horizon, weighting matrices, and constraints. BO then enables systematic tuning by considering the closed-loop system performance as a black-box objective function \cite{chu2020parameter,balakrishna2021learning}. The BO’s surrogate modeling can also be used to capture and update system uncertainties in real time, enabling adaptive MPC strategies for nonlinear or time-varying systems \cite{ong2017gaussian}. Authors in \cite{HIRT2024208} proposed a safe and stability-informed BO for MPC cost function learning, in which a parameterized MPC scheme is safely adjusted to deliver the best closed-loop performance in the presence of a model mismatch between the MPC model and the real plant. A high-dimensional BO framework for sample-efficient MPC tuning was also proposed recently in \cite{KUDVA2024458}.

Despite recent efforts in utilizing BO for control design purposes, in the context of multi-agent systems, the combined BO-MPC control approaches above do not account for the coupling and the interactions between different agents, i.e., through a DMPC scheme. To address that, in this paper, we propose to learn a DMPC scheme using a Multi-Agent BO (MABO) framework aiming at improving the optimal closed-loop performance for each local MPC scheme in the presence of model mismatch. We formulate the fusion of DMPC and multi-agent Markov Decision Processes (MDPs) so that the local MPC schemes obtained from a dual-decomposition method can capture the local value functions associated with the multi-agent MDP. In the proposed learning-based DMPC, we first show how the local cost functions associated with the local MPC schemes can be modified to deliver local optimal policies. We then propose to practically learn a parameterized DMPC by a coordinated learning mechanism in the ADMM-based MABO framework.

\textbf{Major contributions of this work are as follows:}
\begin{itemize}
    \item \textbf{{We propose a modified DMPC scheme and demonstrate its ability to capture the true multi-agent MDP, even when the local MPC schemes, formulated via the dual decomposition method, are based on imperfect models of the real multi-agent system. We provide formal proofs to establish that the modified DMPC scheme serves as a valid approximator for the local optimal policies.}}
    \item \textbf{{We introduce a parametric variant of the modified DMPC scheme and develop an ADMM-based MABO algorithm to enable coordinated learning of the parameterized local MPC schemes.}}
    \item \textbf{{We conduct a rigorous convergence analysis of the proposed coordination-based distributed learning algorithm, supported by formal proofs of key properties of the acquisition function employed in the ADMM-based MABO method.}}   
\end{itemize}

This paper is organized as follows. In Section \ref{sec:2}, a distributed MPC scheme based on a dynamic dual-decomposition method is described. The fusion of DMPC and multi-agent MDPs is detailed in Section \ref{sec:3}. In Section \ref{sec:4}, we provide the proposed coordinated BO algorithm for learning the DMPC schemes. Numerical examples are then given in Section \ref{sec:5} to show the accuracy, efficacy, and performance of the proposed MABO-DMPC. Concluding remarks are finally given in Section \ref{sec:6}. 



\section{Preliminaries and Problem Setup}\label{sec:2}

In this section, we first provide the formulation for the centralized control problem and then describe the use of dual decomposition to break the centralized problem into local (distributed) optimization problems. 
We consider a distributed networked optimal control system comprising a set of $\mathcal{M}$ interconnected dynamical subsystems, each denoted as $\mathcal{M} = \{\Sigma_1, \dots, \Sigma_{|\mathcal{M}|}\}$, where $m:=|\mathcal{M}|$ represents the cardinality of the set $\mathcal{M}$. For each subsystem $\Sigma_i$, we define subsystem $\Sigma_j$ as a neighbor of $\Sigma_i$ if the two subsystems are subject to coupled constraints and/or a coupled cost function. We also consider a Distributed Control System (DCS) defined by a graph $\mathcal{G}(\mathcal{M}, \mathcal{E}_c)$, consisting of a set of nodes $\mathcal{M}$ corresponding to the subsystems and edges $\mathcal{E}_c$ representing the interconnections of the subsystems through a coupling cost function. The set $\mathcal{M}_i := \{\Sigma_j \,|\, (i, j) \in \mathcal{E}_c\text{ or } (j, i) \in \mathcal{E}_c, \, i \neq j\}$ then represents the set of subsystems $\Sigma_j$ that are interconnected with $\Sigma_i$. 

\subsection{Centralized Control Problem}

Let the states and inputs of the agent $i$ be denoted by $\vect x_i$ and $\vect u_i$, respectively. We consider a deterministic model of each agent as $\vect{x}_i^{k+1} = \vect{f}_i(\vect{x}_i^k,\,\vect{u}_i^k)$, where $\vect x_i^k\in\mathbb{R}^{n_{x_i}},\vect u_i^k\in\mathbb{R}^{n_{u_i}}$. Let $\vect x^k=\mathrm{col}\left\{\vect x_1^k,\ldots,\vect x_m^k\right\}$ and $\vect u^k=\mathrm{col}\left\{\vect u_1^k,\cdots,\vect u_m^k\right\}$ be the augmented state and control input vectors of multi-agent system, respectively. The corresponding dimensions then are $\vect x^k\in\mathbb{R}^{n_x},\vect u^k\in\mathbb{R}^{n_u}$, $n_x=\sum_{i=1}^m n_{x_{i}},n_u=\sum_{i=1}^m n_{u_{i}}$. A networked control scheme can be based on a centralized optimization problem or a set of local problems that need to be solved at each time instant. The centralized optimization problem is formulated as
\begin{subequations}\label{eq:centralized}
		\begin{align}
		&\min_{\hat{\vect x},\hat {\vect u}} \sum_{i=1}^{m}\Bigg\{T_i\left(\hat {\vect x}_i^{k+N},\hat {\vect w}_i^{k+N}\right)+\sum_{\ell=k}^{k+N-1}l_i\left(\hat {\vect x}_i^{\ell},\hat {\vect w}_i^{\ell},\hat {\vect u}_i^{\ell}\right)\Bigg\} \\
			\mathrm{s.t.}&\quad \hat {\vect x}_i^{\ell+1}=\vect f_i\left(\hat {\vect x}_i^{\ell},\hat {\vect u}_i^{\ell}\right),\quad \hat {\vect x}_i^{k}=\vect s_i^{k},\\
			&\quad \vect h_i\left(\hat {\vect x}_i^{\ell},\hat {\vect u}_i^{\ell}\right)\leq 0,\quad\vect h_i\left(\hat {\vect x}_i^{k+N}\right)\leq 0,\\
			&\quad \hat {\vect w}_i^{\ell}=W_{ij}\left(\hat {\vect x}_j^{\ell}\right),\quad \hat {\vect w}_i^{k+N}=W_{ij}\left(\hat {\vect x}_j^{k+N}\right),\label{eq:coupling}
		\end{align}
\end{subequations}
where $W_{ij}\left(\hat {\vect x}_j^{\ell}\right)_{(j,i)\in\mathcal{E}_c}\in\mathbb{R}^{\Sigma_{j}n_j}$ denotes the tuple of the state vector of all subsystems that influence subsystem $i$, $N$ is the prediction horizon, and $T_i$, $l_i$, $h_i$, and $g_i$ denote the respective terminal cost, stage cost, mixed inequality constraint, and input inequality constraint for agent $i$, respectively.

\begin{Assumption}\label{assump0}
In this paper, we assume that the subsystems in the multi-agent system are of the same order (state-space dimension) while their dynamics can be different, yielding a heterogeneous multi-agent system.
\end{Assumption}
\begin{Remark}\label{rem_1}
In this paper, we consider multi-agent systems with coupled subsystems such that each agent $\Sigma_i$ is affected by all other agents $\mathcal{M}_i$. According to Assumption \ref{assump0}, we then have that $\hat {\vect w}_i^{\ell}\in\mathbb{R}^{(m-1)n_{x_i}}$ in \eqref{eq:coupling}.
\end{Remark}
Solving \eqref{eq:centralized} gives a sequence of optimal input predictions and corresponding state predictions as
\begin{align}
    &\vect{\hat{u}}^\star = \{(\vect{\hat{u}}_1^{k:k+N-1})^\star,\ldots, (\vect{\hat{u}}_m^{k:k+N-1})^\star\},\\\nonumber
    &\vect{\hat{x}}^\star = \{(\vect{\hat{x}}_1^{k:k+N})^\star,\ldots, (\vect{\hat{x}}_m^{k:k+N})^\star\},
\end{align}
where the first element $(\vect{\hat{u}}_i^k)^\star$ of the input sequence $\vect{\hat{u}}_i^\star$ is applied to each agent. At each physical time instant $k$, a new state $\vect{x}_i^k$ is received, and problem \eqref{eq:centralized} is solved again, producing a new $\vect{\hat{u}}_i^\star$ and $(\vect{\hat{u}}_i^k)^\star$ for each agent. However, repeatedly solving the centralized MPC problem $\eqref{eq:centralized}$ can fail for large-scale systems where the communication bandwidth is restricted. To address this issue, we use a Distributed MPC (DMPC) scheme based on dual decomposition. Next, we show how to modify \eqref{eq:centralized} to arrive at a fully distributed formulation.

\subsection{Dynamic Dual Decomposition}

Taking into account the coupling constraints \eqref{eq:coupling} of the centralized optimization problem \eqref{eq:centralized}, one can introduce slack variable $\bar {\vect w}_i^{\ell}$ capturing the effect of other agents on the agent $i$ through $W_{ij}\left(\hat {\vect x}_j^{k+\ell}\right)_{(j,i)\in\mathcal{E}_c}$. The centralized optimization scheme can then be reformulated as
\begin{subequations}\label{eq:centralized_new}
		\begin{align}
		&\max_{\vect\mu}\min_{\hat{\vect x},\hat {\vect u},\bar{\vect w}} \sum_{i=1}^{m}\Bigg\{T_i^{\vect\mu}\left(\hat {\vect x}_i^{k+N},\bar {\vect w}_i^{k+N}\right)\\\nonumber
        &\qquad\qquad\qquad\qquad+\sum_{\ell=k}^{k+N-1}L_i^{\vect\mu}\left(\hat {\vect x}_i^{\ell},\hat {\vect u}_i^{\ell},\bar{\vect w}_i^{\ell}\right)\Bigg\}\\
			\mathrm{s.t.}&\quad \hat {\vect x}_i^{\ell+1}=\vect f_i\left(\hat {\vect x}_i^{\ell},\hat {\vect u}_i^{\ell}\right),\quad \hat {\vect x}_i^{k}=\vect s_i^{k},\\
			&\quad \vect h_i\left(\hat {\vect x}_i^{\ell},\hat {\vect u}_i^{\ell}\right)\leq 0,\quad\vect h_i\left(\hat {\vect x}_i^{k+N}\right)\leq 0,
		\end{align}
\end{subequations}
where
\begin{subequations}
\begin{align}
    &T_i^{\vect\mu}=T_i\left(\hat {\vect x}_i^{k+N},\bar {\vect w}_i^{k+N}\right)\\\nonumber
&+\left(\vect\mu_i^{k+N}\right)^\top\bar{\vect w}_i^{k+N}-\sum_{j=1,j\neq i}^{m}\left(\vect\mu_{ji}^{k+N}\right)^\top W_{ji}\left(\hat {\vect x}_i^{k+N}\right),\\
    &L_i^{\vect\mu}=l_i\left(\hat {\vect x}_i^{\ell},\hat {\vect u}_i^{\ell},\bar{\vect w}_i^{\ell}\right)\\\nonumber
&+\left(\vect\mu_i^{\ell}\right)^\top\bar{\vect w}_i^{\ell}-\sum_{j=1,j\neq i}^{m}\left(\vect\mu_{ji}^{\ell}\right)^\top W_{ji}\left(\hat {\vect x}_i^{\ell}\right),
\end{align}
\end{subequations}
and $\vect\mu_i^{\ell}=\left(\vect\mu_{ij}^\ell\right)_{(j,i)\in\mathcal{E}_c}\in\mathbb{R}^{\Sigma_{j}n_j}$. The local MPC problem is then formulated as
\begin{subequations}\label{eq:DMPC}
		\begin{align}
		    &\left(\hat {\vect x}_i^{(I)}[k:k+N],\hat {\vect u}_i^{(I)}[k:k+N-1],\bar {\vect w}_i^{(I)}[k:k+N]\right)\nonumber\\
		    &\qquad=\argmin_{\hat {\vect x}_i,\hat {\vect u}_i,\bar {\vect w}_i}~~ \Bigg\{T_i^{\vect\mu^{(I)}}\left(\hat {\vect x}_i^{k+N},\bar {\vect w}_i^{k+N}\right)\\\nonumber
        &\qquad\qquad\qquad\qquad\qquad+\sum_{\ell=k}^{k+N-1}L_i^{\vect\mu^{(I)}}\left(\hat {\vect x}_i^{\ell},\hat {\vect u}_i^{\ell},\bar{\vect w}_i^{\ell}\right)\Bigg\}  \\
			&\mathrm{s.t.}\quad \hat {\vect x}_i^{\ell+1}=\vect f_i\left(\hat {\vect x}_i^{\ell},\hat {\vect u}_i^{\ell}\right),\quad \hat {\vect x}_i^{k}=\vect s_i^{k},\\
			&\quad \vect h_i\left(\hat {\vect x}_i^{\ell},\hat {\vect u}_i^{\ell}\right)\leq 0,\quad\vect h_i\left(\hat {\vect x}_i^{k+N}\right)\leq 0,
		\end{align}
\end{subequations}
and the local multipliers are updated as
\begin{align}\label{eq:update_mu}
    \vect\mu_{i}^{(I+1)}=\vect\mu_{i}^{(I)}+\beta_i\left(\bar{\vect w}_{i}^{(I)}-\hat {\vect w}_i^{(I)}\right),
\end{align}
where $\beta_i>0$ is the step size. The stopping criteria are then met ($\vect\mu_i^{I}\rightarrow\vect\mu_i^\star$) if either the number of iterations exceeds $ I\geq I_{\text{max}}$ or one of the following conditions is satisfied:
\begin{align}
 &\left \|\vect\mu_{i}^{(I)}-\vect\mu_{i}^{(I-1)}\right \|_2<\epsilon_1,\quad\left \|\bar{\vect w}_{i}^{(I)}-\bar{\vect w}_{i}^{(I-1)}\right \|_2<\epsilon_2,
\end{align}
for some positive threshold $\epsilon_1,\epsilon_2$. It is worth noting that either one of the two conditions above can be used as a stopping criterion if local MPC models can accurately capture the real subsystems. Consequently, we denote by $\hat{\vect\mu}_i^\star$ the local multipliers ($\vect\mu_i^\star$ when converged) when coping with imperfect MPC models.


\section{Fusion of DMPC and MDP}\label{sec:3}

Let $\vect s_i\in\mathcal{S}_i$ and $\vect a_i\in\mathcal{A}_i$ denote, respectively, the state and action assigned to the agent $i$ where $\mathcal{S}_i$ is the local state space and $\mathcal{A}_i$ is the local action space. We further denote $\vect s=\mathrm{col}\left\{\vect s_1,\cdots,\vect s_m\right\}\in\mathcal{S}$ as the state of the multi-agent system. Similarly, the action of the whole system is defined as $\vect a=\mathrm{col}\left\{\vect a_1,\cdots,\vect a_m\right\}\in\mathcal{A}$ while the joint (global) policy is $\vect\pi=\mathrm{col}\left\{\vect\pi_1,\cdots,\vect\pi_m\right\}$ with the policy of agent $i$ denoted as $\vect\pi_i$. Let us consider $m$ agents acting in a discounted multi-agent Markov Decision Process (MDP), which is defined by a tuple $\left\{\mathcal{S},\mathcal{A},L,\mathbb{P},\gamma\right\}$ where $L(\vect s^k,\vect a^k)$ is the global cost function, $\mathbb{P}[\vect s^{k+1}|\vect s^k,\vect a^k]$ is the true transition model, and $\gamma\in(0,1]$ is the discount factor.

\subsection{Description of the Multi-agent MDP}

In this paper, we consider a class of multi-agent systems in which each agent aims to contribute to the whole system by minimizing a coupling cost $l_i^c$ while its own task to minimize the local cost $l_i$ is preserved. The total cost assigned to each agent is then defined as
\begin{align}
     L_i\left(\vect{\tilde s}^k,\vect a_i^k\right)=l_i\left(\vect s_i^k,\vect a_i^k\right)+l_i^c\left(\vect{\tilde s}^k\right),
\end{align}
where the augmented state $\vect{\tilde s}^k$ is defined as $\vect{\tilde s}^k=\left[\vect s_i^k,W_{ij}\left(\vect s_j^k\right)\right]^\top$. Notice that the coupling cost $l_i^c$ is defined such that the effect of the coupling constraint \eqref{eq:coupling} on the centralized problem \eqref{eq:centralized} can be captured during the learning process, as detailed in Section \ref{sec:4}. From the perspective of each agent $i$, an MDP is then defined by a tuple $\left\{\mathcal{S},\mathcal{A}_i,L_i,\mathbb{P}_i,\gamma\right\}$ with the following transition model 
\begin{align}
    \mathbb{P}_i[\vect{\tilde s}^{k+1}|\vect{\tilde s}^{k},\vect a_i]:=\mathbb{P}\Big[\vect{\tilde s}^{k+1}|\vect{\tilde s}^k,\vect a_i^k,\vect\pi_{j\in\mathcal{M}_i}\left(\vect{\tilde s}^{k}\right)\Big].
\end{align}
 Let us define the value function of the multi-agent system as
\begin{align}\label{eq:central_V}
  V^{\vect\pi}\left(\vect{\tilde s}^k\right)=\mathbb{E}\left[\sum_{\ell=k}^\infty\gamma^{\ell-k} L\left(\vect{\tilde s}^\ell,\vect a^\ell\right)\Bigg|\vect a^\ell=\vect \pi(\vect{\tilde s}^\ell)\right],
\end{align} 
where the global stage cost function $L$ reads as
\begin{align}
     L\left(\vect{\tilde s}^\ell,\vect a^\ell\right)=\sum_{i=1}^m L_i\left(\vect{\tilde s}^\ell,\vect a_i^\ell\right).
\end{align}
The expectation $\mathbb{E}$ in \eqref{eq:central_V} is taken over the distribution of the multi-agent Markov chain resulting from the closed-loop system with the joint policy ${\vect{ \pi}}$. 
\begin{Remark}\label{Value_decomposition}
To decompose the multi-agent value function \eqref{eq:central_V}, one can consider a linear value-decomposition as \cite{DBLP}
\begin{align}\label{eq:v_decom}
    &V^{\vect\pi}\left(\vect{\tilde s}^k\right)=\mathbb{E}\left[\sum_{\ell=k}^\infty\gamma^{\ell-k} L_1\left(\vect{\tilde s}^\ell,\vect a_1^\ell\right)\Bigg|\vect a^\ell=\vect \pi(\vect{\tilde s}^\ell)\right]\\\nonumber
    &\qquad\quad+\mathbb{E}\left[\sum_{\ell=k}^\infty\gamma^{\ell-k} L_2\left(\vect{\tilde s}^\ell,\vect a_2^\ell\right)\Bigg|\vect a^\ell=\vect \pi(\vect{\tilde s}^\ell)\right]\\\nonumber
    &\qquad\quad+\cdots+\mathbb{E}\left[\sum_{\ell=k}^\infty\gamma^{\ell-k} L_m\left(\vect{\tilde s}^\ell,\vect a_m^\ell\right)\Bigg|\vect a^\ell=\vect \pi(\vect{\tilde s}^\ell)\right]\\\nonumber
    &\qquad=: V_1^{\vect\pi}\left(\vect{\tilde s}^k\right)+V_2^{\vect\pi}\left(\vect{\tilde s}^k\right)+\cdots+V_m^{\vect\pi}\left(\vect{\tilde s}^k\right)\\\nonumber
    &\qquad\qquad\qquad\qquad\qquad\qquad\qquad\qquad=\sum_{i=1}^m V_i^{\vect\pi}\left(\vect{\tilde s}^k\right).
\end{align}
\end{Remark}

\subsection{Modifying an Approximate Multi-agent MDP}

Considering the value decomposition in \eqref{eq:v_decom}, the $N$-step local value function then is
\begin{align}
    &V_i^{N,\vect\pi}\left(\vect{\tilde s}^k\right)=\\\nonumber
    &\mathbb{E}\left[\gamma^N V_i^f\left(\vect {\tilde s}^{k+N}\right)+\sum_{\ell=k}^{k+N-1}\gamma^{\ell-k} L_i\left(\vect {\tilde s}^\ell,\vect \pi_i\left(\vect {\tilde s}^\ell\right)\right)\right],
\end{align}
where $\vect{\tilde s}^k=\left[\vect s_i^k,W_{ij}\left(\vect s_j^k\right)\right]^\top$. By modifying the local terminal cost $V_i^f$ and the local stage cost $L_i$, we next show that the value function above can capture the local value function $V_i^{\vect\pi}\left(\vect{\tilde s}^k\right)$ even if it is constructed based on an $N$-step definition, and trajectories are imperfectly generated based on an inaccurate transition model $\hat{\mathbb{P}}_i[\vect{\hat s}^{k+1}|\vect s^k,\vect a_i^k]$.

\begin{Assumption}\label{assump1}
We assume that the following set is non-empty:
\begin{align}
    &\Xi=:\\\nonumber
    &\left\{\vect{\hat s}\in\mathcal{S}\Bigg|\left|\mathbb{E}\left[V_i^{\vect\pi}\left(\vect{\hat s}^\ell\right)\right]\right|<\infty,\forall \ell\in\mathbb{N}_{\geq k},\forall i\in \left\{1,\ldots,m\right\}\right\}
\end{align}
where $\vect{\hat s}^\ell=\left[\vect {\hat s}_i^\ell,W_{ij}\left(\vect {\hat s}_j^\ell\right)\right]^\top$ and $\mathbb{N}_{\geq k}$ denotes the set of natural numbers greater than $k-1$.
\end{Assumption}

\begin{theorem}\label{theorem:MPC-V}
Under Assumption \ref{assump1}, one can find a modified local terminal cost $\hat{V}_i^f$ and modified local stage cost $\hat L_i$ such that
\begin{align}\label{eq::modif-value}
    &\hat V_i^{N,\vect\pi}\left(\vect{\tilde s}^k\right)=\\\nonumber
    &\mathbb{E}\left[\gamma^N \hat V_i^f\left(\vect{\hat s}^{k+N}\right)+\sum_{\ell=k}^{k+N-1}\gamma^{\ell-k} \hat L_i\left(\vect{\hat s}^\ell,\vect \pi_i\left(\vect{\hat s}^\ell\right)\right)\right]
\end{align}
delivers the same local value function as one associated with the true multi-agent MDP on $\Xi$ and for any $N$:
\begin{align}
    \hat V_i^{N,\vect\pi}\left(\vect{\tilde s}^k\right)=V_i^{\vect\pi}\left(\vect{\tilde s}^k\right).
\end{align}
\end{theorem}
\begin{proof}
Let us choose the modified local terminal and stage costs as
\begin{subequations}\label{eq::modif-costs}
   \begin{align}
       &\hat V_i^f\left(\hat{\vect s}^{k+N}\right)=V_i^{\vect\pi}\left(\hat{\vect s}^{k+N}\right),\\\label{eq::modif-cost}
       &\hat L_i\left(\hat{\vect s}^\ell,\vect \pi_i\left(\hat{\vect s}^\ell\right)\right)=\\\nonumber
       &\left\{\begin{matrix}
          V_i^{\vect\pi}\left(\hat{\vect s}^\ell\right)-\gamma\mathcal{V}_i^{+}\left(\hat{\vect s}^\ell,\vect a_i^\ell\right)\quad\text{if}\quad\left|\mathcal{V}_i^{+}\left(\hat{\vect s}^\ell,\vect a_i^\ell\right)\right|<\infty\\ 
          \infty\qquad\qquad\qquad\qquad\text{Otherwise}
         \end{matrix}\right.,
   \end{align}
\end{subequations}
where $\mathcal{V}_i^{+}\left(\hat{\vect s}^\ell,\vect a_i^\ell\right)=\mathbb{E}\left[V_i^{\vect\pi}\left(\hat{\vect s}^{\ell+1}\right)\Bigg|\vect a^\ell=\vect \pi(\vect s^\ell)\right]$. Under Assumption 1, the modified local terminal and stage costs in \eqref{eq::modif-costs} are retained finite, and using a telescoping sum, the centralized $N$-step value function \eqref{eq::modif-value} is rewritten as
\begin{align}
    &\hat V_i^{N,\vect\pi}\left(\vect {\tilde s}^k\right)=
    \mathbb{E}\Bigg[\gamma^{N}V_i^{\vect\pi}\left(\hat{\vect s}^{k+N}\right) +V_i^{\vect\pi}\left(\hat{\vect s}^k\right)\\\nonumber
    &\qquad-\gamma\mathbb{E}\left[V_i^{\vect\pi}\left(\hat{\vect s}^{k+1}\right)\right]+\gamma V_i^{\vect\pi}\left(\hat{\vect s}^{k+1}\right)\\\nonumber
    &\qquad-\gamma^2\mathbb{E}\left[V_i^{\vect\pi}\left(\hat{\vect s}^{k+2}\right)\right]+\gamma^2V_i^{\vect\pi}\left(\hat{\vect s}^{k+2}\right)+\ldots\\\nonumber
    &+\gamma^{N-1}V_i^{\vect\pi}\left(\hat{\vect s}^{k+N-1}\right)-\gamma^N\mathbb{E}\left[V_i^{\vect\pi}\left(\hat{\vect s}^{k+N}\right)\right]\Bigg]\\\nonumber
    &=V_i^{\vect\pi}\left(\left[\vect {\hat s}_i^k,W_{ij}\left(\vect {\hat s}_j^k\right)\right]^\top\right)\Big|_{\vect{\hat s}^k=\vect{\tilde s}^k}\stackrel{\text{(Remark 1)}}{=}V_i^{\vect\pi}\left(\vect{\tilde s}^k\right).
\end{align}
\end{proof}

\begin{theorem}
For local value functions $\hat V_i^{N,\vect\pi}\left(\vect{\tilde s}^k\right)$ learned on the transition model $\hat{\mathbb{P}}_i\Big[\vect{\tilde s}^{k+1}|\vect{\tilde s}^k,\vect\pi_i,\vect\pi_{j\in\mathcal{M}_i}^\star\Big]$, the corresponding local policy $\vect\pi_i$ then converges to the optimal local policy $\vect\pi_i^\star$.
\end{theorem}
\begin{proof}
Let us define the Bellman equation as
    \begin{align}
        &\hat V_i^{N,\vect\pi}\left(\vect{\tilde s}^k\right)=L_i\left(\vect{\tilde s}^k,\vect\pi_i\left(\vect{\tilde s}^k\right)\right)\\\nonumber
    &\qquad\qquad+\gamma\mathbb{E}_{\hat{\mathbb{P}}_i\Big[\vect{\tilde s}^{k+1}|\vect{\tilde s}^k,\vect\pi_i,\vect\pi_{j\in\mathcal{M}_i}\Big]}\left[\hat V_i^{N,\vect\pi}\left(\vect{\tilde s}^{k+1}\right)\right].
    \end{align}
From the perspective of other agents $\mathcal{M}_i$, the optimal version of the above Bellman equation then reads as
    \begin{align}
        &\min_{\vect\pi_{j\in\mathcal{M}_i}}\hat V_i^{N,\vect\pi}\left(\vect{\tilde s}^k\right)=L_i\left(\vect{\tilde s}^k,\vect\pi_i\left(\vect{\tilde s}^k\right) \right)\\\nonumber
    &+\gamma\min_{\vect\pi_{j\in\mathcal{M}_i}}\mathbb{E}_{\hat{\mathbb{P}}_i\Big[\vect{\tilde s}^{k+1}|\vect{\tilde s}^k,\vect\pi_i,\vect\pi_{j\in\mathcal{M}_i}\Big]}\left[\min_{\vect\pi_{j\in\mathcal{M}_i}}\hat V_i^{N,\vect\pi}\left(\vect{\tilde s}^{k+1}\right)\right].
    \end{align}
    By folding $\min_{\vect\pi_{j\in\mathcal{M}_i}}$ into $\hat{\mathbb{P}}_i$, we then obtain
    \begin{align}\label{eq:Bell_j}
                &\min_{\vect\pi_{j\in\mathcal{M}_i}}\hat V_i^{N,\vect\pi}\left(\vect{\tilde s}^k\right)=L_i\left(\vect{\tilde s}^k,\vect\pi_i\left(\vect{\tilde s}^k\right) \right)\\\nonumber
    &\qquad+\gamma\mathbb{E}_{\hat{\mathbb{P}}_i\Big[\vect{\tilde s}^{k+1}|\vect{\tilde s}^k,\vect\pi_i,\vect\pi_{j\in\mathcal{M}_i}^\star\Big]}\left[\min_{\vect\pi_{j\in\mathcal{M}_i}}\hat V_i^{N,\vect\pi}\left(\vect{\tilde s}^{k+1}\right)\right].
    \end{align}
We then minimize \eqref{eq:Bell_j} w.r.t $\vect\pi_i$ as
\begin{align}\label{eq:Bell_ij}
            &\min_{\vect\pi_i}\min_{\vect\pi_{j\in\mathcal{M}_i}}\hat V_i^{N,\vect\pi}\left(\vect{\tilde s}^k\right)=L_i\left(\vect{\tilde s}^k,\vect\pi_i\left(\vect{\tilde s}^k\right) \right)\\\nonumber
    &+\gamma\mathbb{E}_{\hat{\mathbb{P}}_i\Big[\vect{\tilde s}^{k+1}|\vect{\tilde s}^k,\vect\pi_i,\vect\pi_{j\in\mathcal{M}_i}^\star\Big]}\left[\min_{\vect\pi_i}\min_{\vect\pi_{j\in\mathcal{M}_i}}\hat V_i^{N,\vect\pi}\left(\vect{\tilde s}^{k+1}\right)\right],
\end{align}
and from Theorem \ref{theorem:MPC-V}, the following equality then holds
\begin{align}\label{eq:opt_value}
\min_{\vect\pi_i}\min_{\vect\pi_{j\in\mathcal{M}_i}}\hat V_i^{N,\vect\pi}\left(\vect{\tilde s}^k\right)=\min_{\vect\pi_i}\min_{\vect\pi_{j\in\mathcal{M}_i}}V_i^{\vect\pi}\left(\vect{\tilde s}^k\right)=V_i^\star\left(\vect{\tilde s}^k\right).
\end{align}
Using \eqref{eq:opt_value}, we then rewrite \eqref{eq:Bell_ij} as
\begin{align}\label{eq:Bell_opt}
&V_i^\star\left(\vect{\tilde s}^k\right)=L_i\left(\vect{\tilde s}^k,\vect\pi_i^\star\left(\vect{\tilde s}^k\right) \right)\\\nonumber
&\qquad\qquad\qquad+\gamma\mathbb{E}_{\hat{\mathbb{P}}_i\Big[\vect{\tilde s}^{k+1}|\vect{\tilde s}^k,\vect\pi_i,\vect\pi_{j\in\mathcal{M}_i}^\star\Big]}\left[V_i^\star\left(\vect{\tilde s}^{k+1}\right)\right].
\end{align}
Then, if there exists $\bar{\vect\pi}_i\neq\vect\pi_i^\star$ such that $\bar V_i\left(\vect{\tilde s}^k\right)\leq V_i^\star\left(\vect{\tilde s}^k\right)$, it will contradict the Bellman equation \eqref{eq:Bell_opt} so that $\vect\pi^\star$ is not the optimal joint policy. We then have that
\begin{align}\label{eq:local_prob}
    \vect\pi_i^\star\left(\vect{\tilde s}^k\right)\in\mathrm{arg}\min_{\vect\pi_i}\min_{\vect\pi_{j\in\mathcal{M}_i}}\hat V_i^{N,\vect\pi}\left(\vect{\tilde s}^k\right).
\end{align}
\end{proof}

\subsection{Parameterized DMPC for Local Policy Approximation}

In this section, we propose to use a parameterized DMPC scheme as approximators for the local policy $\vect\pi_i,i=1,\ldots,m$. More precisely, we aim to learn a local MPC parameterized by $\vect\theta_i$ such that the parameterized policy $\vect\pi_{\vect\theta_i}$ delivered from each local MPC can provide an accurate approximation of the optimal policy $\vect\pi_i^\star$ at $\vect\zeta^\star:=\left\{\vect\theta_1^\star,\ldots,\vect\theta_m^\star\right\}$ even if the MPC model cannot capture the real system perfectly. By modifying the DMPC scheme, we have
\begin{equation}
    V^{\vect\pi^\star}\left(\vect{\tilde s}^k\right)=\sum_{i=1}^m V_i^{\vect\pi^\star}\left(\vect{\tilde s}^k\right)=V^\star\left(\vect{\tilde s}^k\right)= \sum_{i=1}^m V_i^{\vect\theta_i^\star},
\end{equation}
where $\vect\pi^\star=\mathrm{col}\left\{\vect\pi_1^\star,\cdots,\vect\pi_m^\star\right\}$, and $V_i^{\vect\theta_i^\star}$ is the local optimal value function captured by the parameterized DMPC scheme (as an approximator for $V_i^\star$).

\begin{Remark}
To address the local optimization problem \eqref{eq:local_prob} at every time instant $k$, we use a parameterized DMPC based on dual decomposition. We then propose to parameterize each local MPC scheme so that the local value function $\hat V_i^{N,\vect\pi}\left(\vect{\tilde s}^k\right)$ can be replaced by $V_i^{\vect\theta_i}$ delivered from the local MPC. To capture the local optimal policies $\vect\pi_i^\star\left(\vect{\tilde s}^k\right)$, one needs to learn the corresponding local MPC schemes in a coordinated manner such that $\vect\zeta:=\left\{\vect\theta_1,\ldots,\vect\theta_m\right\}\rightarrow\vect\zeta^\star$ results in
\begin{align}
\hat{\vect\mu}^\star\left(\vect\zeta\right)\rightarrow\vect\mu^\star\left(\vect\zeta^\star\right),\quad V_i^{\vect\theta_i^\star}\left(\vect{\tilde s}^k,\vect\mu^\star\right)=V_i^\star\left(\vect{\tilde s}^k\right),
\end{align}
and $\vect\pi_i^\star\approx\vect\pi_i^{\vect\zeta^\star}\left(\vect{\tilde s}^k,\vect\mu^\star\right)\in\mathrm{arg}\min_{\vect\pi_i}V_i^{\vect\theta_i^\star}\left(\vect{\tilde s}^k,\vect\mu^\star\right)$.
\end{Remark}

Taking into account the previous observation in Remark 3, the value function captured from a parametrized local MPC scheme must be a function of the complete state $\vect{\tilde s}^k$. Although the value function associated with \eqref{eq:DMPC} is implicitly a function of $\vect{\tilde s}^k$ through the multipliers $\vect\mu_i$ and the slack variables $\bar {\vect w}_i^\ell$, it is only a function of the local state $\vect s_i^k$ explicitly. Hence, the same modification as in Theorem \ref{theorem:MPC-V} will not work for the DMPC scheme \eqref{eq:DMPC}. To address this issue, without loss of generality, we propose to force the initial slack variables to be $\bar {\vect w}_i^k=W_{ij}\left( {\vect{s}}_j^{k}\right)$. We also introduce an additional constraint to smooth out the slack variables. This constraint also acts as a model for propagating the effect of $W_{ij}\left( {\vect s}_j^{k}\right)$ along the prediction horizon. The proposed DMPC scheme is then formulated as
\begin{subequations}\label{eq:modif_DMPC}
		\begin{align}
		    &\min_{\hat {\vect x}_i,\hat {\vect u}_i,\bar {\vect w}_i}T_i^{\vect\mu^{(I)}}\left(\hat {\vect x}_i^{k+N},\bar {\vect w}_i^{k+N},\vect\delta_i^{k+N}\right)\\\nonumber
        &\qquad\qquad\qquad\qquad+\sum_{\ell=k}^{k+N-1}L_i^{\vect\mu^{(I)}}\left(\hat {\vect x}_i^{\ell},\hat {\vect u}_i^{\ell},\bar{\vect w}_i^{\ell},\vect\delta_i^{\ell}\right)\\
			&\mathrm{s.t.}\quad \hat {\vect x}_i^{\ell+1}=\vect f_i\left(\hat {\vect x}_i^{\ell},\hat {\vect u}_i^{\ell}\right),\quad \hat {\vect x}_i^{k}=\vect s_i^{k},\\
            &\qquad \bar {\vect w}_i^k=W_{ij}\left({\vect s}_j^{k}\right),\\
             &\qquad\bar {\vect w}_i^{\ell+1}=\bar {\vect w}_i^{\ell} +\vect\delta_i^\ell,\\
			&\qquad \vect h_i\left(\hat {\vect x}_i^{\ell},\hat {\vect u}_i^{\ell}\right)\leq 0,\quad\vect h_i\left(\hat {\vect x}_i^{k+N}\right)\leq 0.
		\end{align}
\end{subequations}
We next show that the value function associated to the proposed DMPC scheme above can capture the true value function $V_i^{\vect\pi}\left(\vect{\tilde s}^k\right)$ even if the underlying local models $\vect f_i$ do not capture the real subsystems (agents). 

\begin{Corollary}\label{theorem:modif-MPC-V}
Considering Remark \ref{rem_1} and under Assumption \ref{assump1}, one can find the modified cost functions $\hat T_i^{\vect\mu^{(I)}}$ and $\hat L_i^{\vect\mu^{(I)}}$ such that the value function associated with \eqref{eq:modif_DMPC} can capture $V_i^{\vect\pi}\left(\vect{\tilde s}^k\right)$ at $\vect\mu^\star$ even if the trajectories are imperfectly evolved by $\vect{\hat f}_i$.
\end{Corollary}

\begin{proof}
Assuming that the trajectories $\hat{\vect s}^1,\ldots,\hat{\vect s}^N$ in \eqref{eq::modif-value} are approximately deterministic, one can adopt the same structure without expectation to define the value function associated with \eqref{eq:modif_DMPC} using a possibly wrong deterministic model. The modified value function is then defined as
\begin{align}\label{eq:v_dmpc}
        &\hat V^{\text{DMPC}}=\gamma^N\hat T_i^{\vect\mu^{\star}}\left(\hat {\vect x}_i^{k+N},\bar {\vect w}_i^{k+N}\right)\\\nonumber
        &\qquad\qquad\qquad\qquad\qquad+\sum_{\ell=k}^{k+N-1}\gamma^{\ell-k}\hat L_i^{\vect\mu^{\star}}\left(\hat {\vect x}_i^{\ell},\hat {\vect u}_i^{\ell},\bar{\vect w}_i^{\ell}\right),
\end{align}
where $\hat T_i^{\vect\mu^{\star}},\hat L_i^{\vect\mu^{\star}}$ are the modified versions of the costs in \eqref{eq:modif_DMPC} with the last argument dropped. We then select the modified terminal and stage costs to be the same as those in Theorem \ref{theorem:MPC-V} as
    \begin{subequations}\label{eq::modif-costs_DMPC}
   \begin{align}
       &\hat T_i^{\vect\mu^{\star}}=V_i^{\vect\pi}\left(\hat {\vect x}_i^{k+N},\bar {\vect w}_i^{k+N}\right),\\
       &\hat L_i^{\vect\mu^{\star}}=V_i^{\vect\pi}\left(\hat {\vect x}_i^{k},\bar {\vect w}_i^{k}\right)-\gamma V_i^{\vect\pi}\left(\hat {\vect x}_i^{k+1},\bar {\vect w}_i^{k+1}\right).
   \end{align}
   \end{subequations}
Using \eqref{eq::modif-costs_DMPC}, we then rewrite the modified value function \eqref{eq:v_dmpc} as
       \begin{align}
        &\hat V^{\text{DMPC}}=\gamma^NV_i^{\vect\pi}\left(\hat {\vect x}_i^{k+N},\bar {\vect w}_i^{k+N}\right)\\\nonumber
        &\quad+\sum_{\ell=k}^{k+N-1}\gamma^{\ell-k}\left(V_i^{\vect\pi}\left(\hat {\vect x}_i^{\ell},\bar {\vect w}_i^{\ell}\right)-\gamma V_i^{\vect\pi}\left(\hat {\vect x}_i^{\ell+1},\bar {\vect w}_i^{\ell+1}\right)\right).
    \end{align}
Using a telescoping sum, we then obtain 
    \begin{align}
         &\hat V^{\text{DMPC}}=\gamma^NV_i^{\vect\pi}\left(\hat {\vect x}_i^{k+N},\bar {\vect w}_i^{k+N}\right)\\\nonumber
         &\qquad+V_i^{\vect\pi}\left(\hat {\vect x}_i^{k},\bar {\vect w}_i^{k}\right)-\gamma^NV_i^{\vect\pi}\left(\hat {\vect x}_i^{k+N},\bar {\vect w}_i^{k+N}\right).
    \end{align}
Under Remark \ref{rem_1} and applying the initial conditions $\hat {\vect x}_i^{k}=\vect s_i^{k}$ and $\bar {\vect w}_i^k=W_{ij}\left({\vect s}_j^{k}\right)$, we then obtain $\hat V^{\text{DMPC}}=V_i^{\vect\pi}\left(\vect{\tilde s}^k\right)$.
\end{proof}
By Theorem \ref{theorem:MPC-V} and Corollary \ref{theorem:modif-MPC-V}, we prove that the DMPC cost function can be modified such that the local value functions associated with the true multi-agent MDP are captured by the modified local MPC schemes. More specifically, the central theorem aims at showing that there exists such a modification and at understanding its structure. However, the proposed modification structure is not tractable for implementation purposes. To address this, we propose to parameterize the corresponding cost terms and learn them in a coordinated manner for the best closed-loop performance. We then propose to parameterize a discounted version of the DMPC scheme \eqref{eq:modif_DMPC} at $\hat{\vect\mu}^\star$ as 
\begin{subequations}\label{eq:DMPC_par}
		\begin{align}\label{eq:totalCost}
		    & V_i^{\vect\theta_i}\left(\vect{\tilde s}^k,\hat{\vect\mu}^\star\left(\vect\zeta\right)\right)=\\\nonumber
            &\min_{\hat {\vect x}_i,\hat {\vect u}_i,\bar {\vect w}_i,\vect\sigma_i} \gamma^N\left(T_i^{\vect\theta_i,\hat{\vect\mu}^{\star}}\left(\hat {\vect x}_i^{k+N},\bar {\vect w}_i^{k+N}\right)+\vect p_f^\top\vect\sigma_i^{k+N}\right)\\\nonumber
        &+\sum_{\ell=k}^{k+N-1}\gamma^{\ell-k}\left(L_i^{\vect\theta_i,\hat{\vect\mu}^{\star}}\left(\hat {\vect x}_i^{\ell},\hat {\vect u}_i^{\ell},\bar{\vect w}_i^{\ell}\right)+\vect p^\top\vect\sigma_i^\ell+\left\|\vect\delta_i^\ell \right\|_M^2\right)\\
			&\mathrm{s.t.}\nonumber\\
			&\quad \hat {\vect x}_i^{\ell+1}=\vect f_i^{\vect\theta_i}\left(\hat {\vect x}_i^{\ell},\hat {\vect u}_i^{\ell}\right),\quad \hat {\vect x}_i^{k}=\vect s_i^{k},\label{ct1}\\
            &\quad \bar {\vect w}_i^k=W_{ij}^{\vect\theta_i}\left( {\vect s}_j^{k}\right),\\
            &\quad\bar {\vect w}_i^{\ell+1}=\bar {\vect w}_i^{\ell} +\vect\delta_i^\ell,\\
            &\quad h_i^{\vect\theta_i}\left(\hat {\vect x}_i^{\ell},\hat {\vect u}_i^{\ell}\right)\leq\vect\sigma_i^{\ell},\quad h_i^{\vect\theta_i}\left(\hat {\vect x}_i^{k+N}\right)\leq\vect\sigma_i^{k+N},
		\end{align}
\end{subequations}
where the inequality constraints are collected by the parametric function $h_i^{\vect\theta_i}$.
To guarantee the recursive feasibility of local MPC schemes \eqref{eq:DMPC}, we consider some slack variables $\vect\sigma_i$ on both the terminal and stage inequality constraints, where the corresponding cost terms in \eqref{eq:totalCost} are penalized with sufficiently large weights $\vect p$ and $\vect p_f$.


\section{Fusion of Multi-agent Bayesian Optimization (MABO) and DMPC}\label{sec:4}

We assume that the state-space model $\vect{x}_i^{k+1} = \vect{f}_i(\vect{x}_i^k,\,\vect{u}_i^k),~i=1,\ldots,m$, cannot capture the true subsystems exactly, such that there is a model mismatch between the real system and the local MPC model each agent employs. To account for this mismatch, we leverage the results detailed in Section \ref{sec:3} to modify the cost function of the local MPC schemes (and possibly their models and constraints) in a coordinated manner leading to the optimal closed-loop performance of each agent despite the model mismatch. To this end, we employ an Alternating Direction Method of Multipliers (ADMM)-based multi-agent Bayesian optimization approach, initially proposed in \cite{KRISHNAMOORTHY20232232}, and infer the optimal parameterization from closed-loop data. 

\subsection{Review of Bayesian Optimization}

Bayesian Optimization (BO) is an efficient approach for optimizing complex, expensive-to-obtain, and noisy black-box functions as
\begin{align}
   \vect{\theta}^\star = \arg \min_{\vect{\theta}\in\mathcal{P}} J(\vect{\theta}), 
\end{align}
where $\vect\theta^\star$ is the global minimizer of an unknown objective function $J(\vect\theta)$ on the parameter space $\mathcal{P}$. In control applications, BO is particularly useful for learning control hyperparameters, e.g., in MPC. However, there are no existing works employing BO for learning networked model-based control schemes, such as distributed MPC, where the local models cannot represent the true multi-agent system. The BO methods use Gaussian Process (GP) surrogate models to learn an approximation of not-explicitly-known functions with respect to some parameters, e.g., the closed-loop performance of a parameterized MPC $J\left(\vect\theta\right):=J\left(\vect\pi^{\vect\theta}\right)$ with respect to the policy parameters $\vect\theta$. Considering \( J: \mathbb{R}^{n_{\vect\theta}} \rightarrow \mathbb{R} \), we then define the corresponding GP model as
\begin{align}
  J(\vect\theta) \sim \mathcal{GP}(\mu(\vect\theta), k(\vect\theta, \vect\theta')),
\end{align}
where $\mu(\vect{\theta}) = \mathbb{E}[J(\vect{\theta})]$ is the mean function, typically assumed to be zero, and the kernel $k(\vect \theta, \vect \theta')$ is the covariance function that relates points in the input space. 
A commonly used kernel is the squared exponential function as
\begin{align}
    k(\vect{\theta}, \vect{\theta'}) = \sigma_J^2 \exp\left(-\frac{1}{2l^2} ||\vect{\theta} - \vect{\theta'}||^2\right),
\end{align}
where $\sigma_J^2$ is the signal variance, and $l$ is the length scale. The choice of kernel significantly impacts the behavior of the GP model. Given observed data \(\mathcal{D}^k = \{(\vect{\theta}^i, y^i)\}_{i=1}^k,\quad y^i = J(\vect{\theta}^i) + \epsilon^i\) and \(\epsilon^i \sim \mathcal{N}(0, \sigma_n^2)\), we then have
\[
\mathbf{y} \sim \mathcal{N}(\mu(\vect{\Theta}), K + \sigma_n^2 I),
\]
where $\vect\Theta=\left\{\vect\theta^i\right\}_{i=1}^k$ and $K$ denotes the covariance matrix ($n\times n$ kernel matrix with elements $[K]_{(i,j)}=k(\vect\theta_i,\vect\theta_j)$) computed over the training inputs. The posterior distribution of \(J(\vect{\theta}^\star)\) at a new input \(\vect{\theta}^\star\) is then given by
\[
\begin{bmatrix}
\mathbf{y} \\
J^\star 
\end{bmatrix}
\sim \mathcal{N}\left(\begin{bmatrix}
\mu(\vect{\Theta}) \\
\mu(\vect{\theta}^\star)
\end{bmatrix}, \begin{bmatrix}
K + \sigma_n^2 I & k(\vect{\theta}^\star, \vect{\Theta}) \\
k(\vect{\Theta}, \vect{\theta}^\star) & k(\vect{\theta}^\star, \vect{\theta}^\star)
\end{bmatrix}\right).
\]
This allows us to derive the predictive mean and variance for $J(\vect{\theta}^\star)$ as
\[
\mu(\vect{\theta}^\star) = k(\vect{\theta}^\star, \vect{\Theta})(K + \sigma_n^2 I)^{-1} \mathbf{y},
\]
\[
\sigma^2(\vect{\theta}^\star) = k(\vect{\theta}^\star, \vect{\theta}^\star) - k(\vect{\theta}^\star, \vect{\Theta})(K + \sigma_n^2 I)^{-1}k(\vect{\Theta}, \vect{\theta}^\star).
\]
The updated GP models are then used to induce an acquisition function $\alpha(\vect{\theta})$, which helps find the optimum of the unknown objective function $J(\vect\theta)$, that is
\begin{align}
   \vect{\theta}^\star = \arg \min_{\vect{\theta}\in\mathcal{P}} \alpha(\vect{\theta}).
\end{align}
In this paper, we use the Expected Improvement (EI) as an acquisition function:
\begin{align}\label{eq:EI}
    \alpha_{\text{EI}}(\vect{\theta}) = \mathbb{E}[\min(J(\vect{\theta})-J^\star, 0)],
\end{align}
where $J^\star$ is the best observed function value. This acquisition function can be computed using the properties of the Gaussian distribution as
\begin{align}\label{eq:acq_EI}
  \alpha_{\text{EI}}(\vect{\theta}) =-\left( \left(J^\star - \mu(\vect{\theta})\right) \Phi(Z) + \sigma(\vect{\theta}) \phi(Z)\right),  
\end{align}
where $Z = \frac{J^\star - \mu(\vect{\theta})}{\sigma(\vect{\theta})}$, and $\Phi(Z)$ and $\phi(Z)$ denote the standard normal Cumulative Density Function (CDF) and Probability Density Function (PDF) of the standard normal distribution, respectively. As discussed previously, the aim is to leverage BO to learn a parametric distributed MPC so that local MPCs can deliver the corresponding optimal policies $\vect\pi^{\vect\theta_i^\star}$. Therefore, the BO-based learning method must be able to capture the optimal parameters $\vect\theta^\star$. However, popular acquisition functions in the open literature have a greedy and myopic structure that selects the next best point based solely on its immediate impact without considering future rewards. This then causes the process to disregard the long-term effect of the next sample, potentially preventing the learning process from obtaining $\vect\theta^\star$. To address this issue, non-myopic BO methods have been formulated in the MDP framework where they can frame the exploration-exploitation problem as a balance of immediate and future
rewards \cite{nonmyopic_1,nonmyopic_2,nonmyopic_3}.

\begin{Proposition}[BO as Finite-Horizon MDP]\label{prop:nonmyopicAF}
Bayesian optimization (BO) with the Expected Improvement (EI) acquisition function can be formulated as a finite-horizon Markov Decision Process (MDP) of horizon $H$, where the optimal query point $\vect\theta^\star$ minimizes the non-myopic acquisition function
    \begin{align}
        &\tilde\alpha_{\text{EI}}(\vect{\theta})=\alpha_{\text{EI}}(\vect{\theta})\qquad\qquad\text{if}\quad H=1,\\\nonumber
        &\tilde\alpha_{\text{EI}}(\vect{\theta})=V^{H,\vect\tau}\left(\mathcal{D}^k\right)\qquad \text{if}\quad H>1,
    \end{align}
with $V^{H,\vect\tau}\left(\mathcal{D}^k\right)$ denoting the associated value function under policy $\vect\tau: \mathcal{D}^k\rightarrow\vect\theta^{k+1}$.
\end{Proposition}
\begin{proof}
Let us view the BO as an $H$-step Dynamic Programming (DP) by casting it as a finite-horizon MDP for which the state and action are $\mathcal{D}^k$ and $\vect\theta^{k+1}$, respectively. Considering  \(y^i = J(\vect{\theta}^i) + \epsilon^i\) and \(\epsilon^i \sim \mathcal{N}(0, \sigma_n^2)\), the collected data is then a Markov process, and hence the optimal policy $\vect\tau^\star: \mathcal{D}^k\rightarrow\vect\theta^{\star,k+1}$ must satisfy the Bellman principle of optimality \cite{paulson2024}. Let us define a stage cost function as
    \begin{align}
         L\left(\mathcal{D}^\ell,\vect\theta^{\ell+1}\right)=\min(y(\vect{\theta}^{\ell+1})-y^{\star,\ell}, 0),
    \end{align}
where $y^{\star,\ell}$ denotes the minimum observed value in the set $\mathcal{D}^\ell$, i.e., $y^{\star,\ell}=\min\left\{y^0,\ldots,y^\ell\right\}$. We then define the corresponding finite-horizon value function as
\begin{align}
V^{H,\vect\tau}\left(\mathcal{D}^k\right)=&\mathbb{E}_y\left[\sum_{\ell=k}^{k+H} L\left(\mathcal{D}^\ell,\vect\tau^\ell\right)\right].
\end{align}
The Bellman equation then reads as
\begin{align}
V^{H,\vect\tau}\left(\mathcal{D}^k\right)=\mathbb{E}_y\left[L\left(\mathcal{D}^k,\vect\tau^k\right)\right]+\mathbb{E}_y\left[V^{H,\vect\tau}\left(\mathcal{D}^+\right)\right].
\end{align}
We observe that 
\begin{align}
\mathbb{E}_y\left[L\left(\mathcal{D}^k,\vect\tau^k\right)\right]=\mathbb{E}_y\left[\min(y(\vect{\theta}^{\ell+1})-y^{\star,\ell}, 0)\right]=\alpha_{\text{EI}}(\vect{\theta}),
\end{align}
such that the expected improvement acquisition function then satisfies the Bellman equation as
\begin{align}\label{eq_bell_nonmyopic}
V^{H,\vect\tau}\left(\mathcal{D}^k\right)=\alpha_{\text{EI}}(\vect{\theta})+\mathbb{E}_y\left[V^{H,\vect\tau}\left(\mathcal{D}^+\right)\right].
\end{align}
For $H>1$, we then obtain the optimal solution by minimizing the non-myopic acquisition function as 
\begin{align}
\vect\theta^\star=\vect\tau\left(\mathcal{D}^k\right)=\argmin_{\vect\theta}V^{H,\vect\tau}\left(\mathcal{D}^k\right).
\end{align}
\end{proof}
It is noted that the second term in \eqref{eq_bell_nonmyopic} can be practically computed using Monte Carlo (MC) simulations as
\begin{align}
    \mathbb{E}_y\left[V^{H,\vect\tau}\left(\mathcal{D}^+\right)\right]\approx\frac{1}{M}\sum_{i=1}^M\sum_{\ell=k}^{k+H-1} \left(y^{\star,\ell}-y^{\ell+1}\right)^+.
\end{align}

\subsection{Multi-Agent Bayesian Optimization}

In this paper, we describe our approach to learn local MPC schemes in a distributed framework using BO. To this end, we leverage the BO in a multi-agent framework by introducing decomposable Multi-agent Bayesian Optimization (MABO). It should be noted that the use of non-myopic acquisition functions, as obtained in Proposition \ref{prop:nonmyopicAF}, can significantly increase the computational complexity. Moreover, the convergence analysis of ADMM-based MABO becomes more challenging when using non-myopic acquisition functions. In this paper, we then consider the myopic acquisition function $\alpha_{\text{EI}}(\vect{\theta})$ to formulate a MABO algorithm. However, one may use the proposed non-myopic acquisition function for more complex objective functions in the proposed MABO since this look-ahead acquisition function $\tilde\alpha_{\text{EI}}(\vect{\theta})$ can be simply replaced by $\alpha_{\text{EI}}(\vect{\theta})$ without changing the main structure of MABO.

\begin{Remark}
As observed in the formulation of the parametric local MPC scheme \eqref{eq:DMPC_par}, each local value function depends on global parameters $\vect\zeta:=\left\{\vect\theta_1,\ldots,\vect\theta_m\right\}$ through optimal multipliers $\hat{\vect\mu}^\star\left(\vect\zeta\right)$ (each agent has full access to all multipliers through a communication network between local agents). Consequently, each local parametric policy $\vect\pi_i^{\vect\zeta}\left(\vect{\tilde s}^k,\hat{\vect\mu}^\star\right)$ delivered by \eqref{eq:DMPC_par} is also affected by the global parameters $\vect\zeta$ and its neighbor policies $\vect\pi_j^{\vect\zeta}\left(\vect{\tilde s}^k,\hat{\vect\mu}^\star\right)$ so that the corresponding local closed-loop control performance reads as $J_i\left(\vect\zeta\right):=J_i\left(\vect\pi^{\vect\zeta}\left(\vect{\tilde s}^k,\hat{\vect\mu}^\star\right)\right)$ with
\begin{align}\label{eq:_J}
J_i\left(\vect\zeta\right)=\mathbb{E}\Bigg[\sum_{k=0}^{\infty}\gamma^kL_i\left(\vect{\tilde s}^k,\vect a_i^k\right)\Bigg|\vect a^k=\vect \pi^{\vect\zeta}(\vect{\tilde s}^k,\hat{\vect\mu}^\star)\Bigg],
\end{align}
where $L_i\left(\vect{\tilde s}^k,\vect a_i^k\right)$ denotes a baseline cost for each agent.
\end{Remark}
An overview of the proposed learning-based DMPC scheme is illustrated in Fig. \ref{scheme}.
\begin{figure}[htbp!]
		\centering
		\includegraphics[width=.9\linewidth]{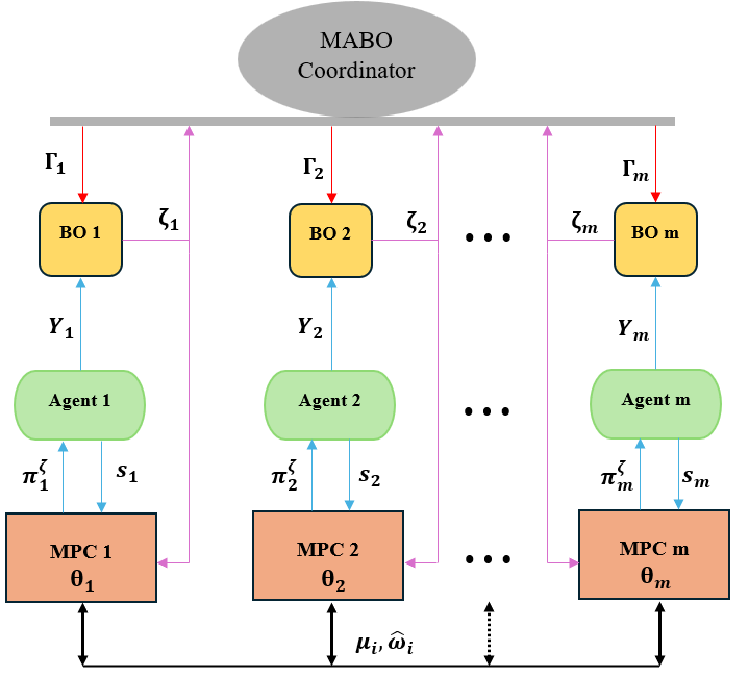}
		\caption{An overview of the proposed MABO-based DMPC.} 
		\label{scheme}
\end{figure}
To employ the closed-loop performance index $J_i\left(\vect\zeta\right)$ as a metric in the context of BO, we practically introduce a finite version of \eqref{eq:_J} as 
\begin{align}
    J_i^N\left(\vect\zeta\right)=\frac{1}{N}\sum_{k=0}^{N}L_i\left(\vect{ \tilde s}^k,\vect\pi_i^{\vect\zeta}\left(\vect{\tilde s}^k,\hat{\vect\mu}^\star\right)\right).
\end{align}
We then implement a coordinated multi-agent learning approach by minimizing the global closed-loop performance as
\begin{align}\label{eq:opt0}
    \min_{\vect\zeta} J^N\left(\vect\zeta\right):=\sum_{i=1}^m J_i^N\left(\vect\zeta\right).
\end{align}
Inspired by \cite{KRISHNAMOORTHY20232232}, we then adopt an ADMM-based MABO to provide a decomposable version of the global optimization problem above aligned with the structure of the proposed parametric DMPC scheme \eqref{eq:DMPC_par}. To this end, let us rewrite the problem \eqref{eq:opt0} as
\begin{subequations}\label{eq:tot_cost_mabo}
\begin{align}
    &\min_{\vect{\bar\zeta},\left\{\vect\zeta_i\right\}\in\mathcal{P}}\sum_{i=1}^mJ_i^N\left(\vect\zeta_i\right)\\
    &\text{s.t.}\quad \vect\zeta_i=\vect{\bar\zeta},\quad\forall i=1,\ldots,m.\label{eq:eq_const}
\end{align}
\end{subequations}
To derive the ADMM scheme for solving the problem above in a distributed manner, we first define the corresponding augmented Lagrangian as
\begin{align}
    \min_{\vect{\bar\zeta},\left\{\vect\zeta_i\right\}\in\mathcal{P}}\sum_{i=1}^mJ_i^N\left(\vect\zeta_i\right)+\vect\lambda_i^\top\left(\vect\zeta_i-\vect{\bar\zeta}\right)+\frac{\rho}{2} \left\|\vect\zeta_i-\vect{\bar\zeta} \right\|^2,
\end{align}
where $\rho>0$ is a constant penalty parameter and $\vect\lambda=\left(\vect\lambda_1,\ldots,\vect\lambda_m\right)$ denotes the set of dual variables (multipliers) associated with the equality constraint \eqref{eq:eq_const}. The local parameters are then obtained by the following local optimization problem
\begin{align}\label{eq:loc_BO}
    \vect\zeta_i^{k+1}=\argmin_{\vect\zeta_i\in\mathcal{P}} J_i\left(\vect\zeta_i\right)+(\vect\lambda_i^k)^\top\vect\zeta_i+\frac{\rho}{2} \left\|\vect\zeta_i-\vect{\bar\zeta}^{k+1} \right\|^2,
\end{align}
where the coordinating variables $\vect \Gamma_i^k:=(\vect{\bar\zeta}^{k+1},\vect\lambda_i^k)$ are updated by the following iterations
\begin{subequations}
\begin{align}
    &\vect{\bar\zeta}^{k+1}=\frac{1}{m}\sum_{i=1}^m\left[\vect\zeta_i^k+\frac{\vect\lambda_i^k}{\rho}\right],\\
    &\vect\lambda_i^{k+1}=\vect\lambda_i^k+\rho\left(\vect\zeta_i^{k+1}-\vect{\bar\zeta}^{k+1}\right).
\end{align}
\end{subequations}
To solve the problem \eqref{eq:loc_BO}, we then minimize the local acquisition function associated with $J_i\left(\vect\zeta_i\right)$ as
\begin{align}\label{eq:local}
    \vect\zeta_i^{k+1}=\argmin_{\vect\zeta_i\in\mathcal{P}_i} \Big(\alpha_{\text{EI}}(\vect\zeta_i)+\Xi_i\left(\vect\zeta_i,\vect \Gamma_i^k\right)\Big),
\end{align}
where the penalty term $\Xi_i$ is
\begin{align}
    \Xi_i\left(\vect\zeta_i,\vect \Gamma_i^k\right)=(\vect\lambda_i^k)^\top\vect\zeta_i+\frac{\rho}{2} \left\|\vect\zeta_i-\vect{\bar\zeta}^{k+1} \right\|^2.
\end{align}
Summary of the proposed multi-agent BO is given in Algorithm 1. The proposed learning-based DMPC is then outlined in Algorithm 2.

\begin{algorithm}
\caption{An ADMM-based multi-agent BO}
\begin{algorithmic}
\STATE \textbf{Input:}\\Parameters $\vect{P}_i^k,i=1,\ldots,m$,\\ Evaluations $\vect Y_i^k,i=1,\ldots,m$,\\ Coordinating variables $\vect \Gamma_i^k:=(\vect{\bar\zeta}^{k+1},\vect\lambda_i^k)$
\STATE \textbf{Output:}\\ $\left(\vect{P}_i^{k+1}, \vect Y_i^{k+1},\vect \Gamma_i^{k+1}\right) =\textbf{\textit{MABO}}\left(\vect{P}_i^{k}, \vect Y_i^{k},\vect \Gamma_i^k\right)$
    \STATE \# In parallel for $i=1,\ldots,m$
    \STATE 1. Fit Gaussian Process models to data $(\vect{P}_i^k, \vect Y_i^k)$
    \STATE 2. Compute $\alpha_{\text{EI}}(\vect\zeta_i)$ for $\vect\zeta_i \in \mathcal{P}_i$ (search space)
    \STATE 3. Compute $\vect{\bar\zeta}^{k+1}$ and $ \Xi_i\left(\vect\zeta_i,\vect \Gamma_i^k\right)$
    \STATE 4. Find $\vect{\zeta}_i^{\text{new}} = \argmin_{\vect\zeta_i\in\mathcal{P}_i} \Big(\alpha_{\text{EI}}(\vect\zeta_i)+\Xi_i\left(\vect\zeta_i,\vect \Gamma_i^k\right)\Big)$
    \STATE 5. Evaluate $y_i^{\text{new}} =J_i^N\left(\vect\zeta^{\text{new}}\right)$
    \STATE 6. Update multipliers $\vect\lambda_i^{k+1}=\vect\lambda_i^k+\rho\left(\vect{\zeta}_i^{\text{new}}-\vect{\bar\zeta}^{k+1}\right)$
    \STATE 7. Update data:\\ $(\vect{P}_i^{k+1}, \vect Y_i^{k+1}) = (\vect{P}_i^k \cup \vect{\zeta}_i^{\text{new}}, \vect Y_i^k \cup y_i^{\text{new}})$
\end{algorithmic}
\end{algorithm}

\begin{algorithm}
\caption{Coordinated learning of the DMPC schemes}
\begin{algorithmic}[1]
\STATE \textbf{Input:}\\
The number of episodes $K$\\
The number of simulation time steps $n$
\FOR{$k = 1$ to $K$}
\STATE Set $\vect s_i^1$ to initial conditions for $i=1,\ldots,m$
    \FOR{$j = 1$ to $n$}  \label{line:simulation_loop}
        \STATE Initialize the multipliers $\vect\mu\rightarrow\mathbf{0}$
        \WHILE{Stopping Criteria are not satisfied}  \label{line:communication_loop}
            \STATE Send/receive $\vect\mu_i$ 
            \STATE MPCs compute the local solutions given $\vect\mu$ 
            \STATE Update $\vect\mu_i,i=1,\ldots,m$ using \eqref{eq:update_mu}
        \ENDWHILE
        \STATE Apply $\vect{a}_i^j=\vect\pi_i^{\vect\theta_i}\left(\vect {\tilde s}^j,\hat{\vect\mu}^\star\right),i=1,\ldots,m$
        \STATE Agents return $\vect{s}_i^{j+1} = \vect{f}_i(\vect{s}_i^j,\,\vect{a}_i^j)$
    \ENDFOR
    \STATE $\left(\vect{P}_i^{k+1}, \vect Y_i^{k+1},\vect \Gamma_i^{k+1}\right) =\textbf{\textit{MABO}}\left(\vect{P}_i^{k}, \vect Y_i^{k},\vect \Gamma_i^k\right)$
    \STATE Update $\vect\theta_i\leftarrow\vect\zeta_i\leftarrow\vect{P}_i^{k+1},i=1,\ldots,m$
\ENDFOR
\end{algorithmic}
\end{algorithm}

\subsection{Convergence Analysis of MABO}
To investigate the convergence of the ADMM-based MABO, we rely on the standard ADMM convergence guarantees while accounting for the properties of the acquisition function $\alpha_{\text{EI}}$. 

\begin{Lemma}
The acquisition function $\alpha_{\text{EI}}(\vect{\zeta}_i)$ defined in \eqref{eq:acq_EI} is differentiable and satisfies Lipschitz continuity such that
\begin{align}
    \|\nabla \alpha_{\text{EI}}(\vect{\zeta}_i) - \nabla \alpha_{\text{EI}}(\vect{\zeta}_j)\| 
    \leq L \|\vect{\zeta}_i - \vect{\zeta}_j\|,
\end{align}
where $L$ is the Lipschitz constant.
\end{Lemma}
\begin{proof}
Consider the acquisition function $\alpha_{\text{EI}}$ as defined in \eqref{eq:acq_EI}. Since $\alpha_{\text{EI}}(\vect{\zeta}_i)$ is a composition of differentiable functions, it is differentiable. The gradient of $\alpha_{\text{EI}}(\vect{\zeta}_i)$ with respect to $\vect{\zeta}_i$ is obtained as
\begin{align}
\nabla \alpha_{\text{EI}}(\vect{\zeta}_i) = \frac{\partial \alpha_{\text{EI}}}{\partial \mu} \nabla \mu(\vect{\zeta}_i) + \frac{\partial \alpha_{\text{EI}}}{\partial \sigma} \nabla \sigma(\vect{\zeta}_i),
\end{align}
where $\frac{\partial \alpha_{\text{EI}}}{\partial \mu} = -\Phi(Z)$ and $\frac{\partial \alpha_{\text{EI}}}{\partial \sigma} = \phi(Z)$. $\nabla \mu(\vect{\zeta}_i)$ and $\nabla \sigma(\vect{\zeta}_i)$ are the gradients of the GP posterior mean and standard deviation, respectively. Both $\nabla \mu(\vect{\zeta}_i)$ and $\nabla \sigma(\vect{\zeta}_i)$ are then bounded due to the smoothness of GP posterior predictions. Therefore, $\nabla \alpha_{\text{EI}}(\vect{\zeta}_i)$ exists and is well-defined. To establish the Lipschitz continuity of $\nabla \alpha_{\text{EI}}(\vect{\zeta}_i)$, we then compute the difference between gradients for two points $\vect{\zeta}_i$ and $\vect{\zeta}_j$ as
\begin{align}
\|\nabla \alpha_{\text{EI}}(\vect{\zeta}_i) - \nabla \alpha_{\text{EI}}(\vect{\zeta}_j)\| 
\leq \|\Delta_\mu\| + \|\Delta_\sigma\|,
\end{align}
where
\[
\Delta_\mu = \left(\frac{\partial \alpha_{\text{EI}}}{\partial \mu} \nabla \mu(\vect{\zeta}_i) - \frac{\partial \alpha_{\text{EI}}}{\partial \mu} \nabla \mu(\vect{\zeta}_j)\right),
\]
\[
\Delta_\sigma = \left(\frac{\partial \alpha_{\text{EI}}}{\partial \sigma} \nabla \sigma(\vect{\zeta}_i) - \frac{\partial \alpha_{\text{EI}}}{\partial \sigma} \nabla \sigma(\vect{\zeta}_j)\right).
\]
Using the smoothness of $\Phi(Z)$ and $\nabla \mu(\vect{\zeta})$, there exists a constant $L_\mu$ such that $\|\Delta_\mu\| \leq L_\mu \|\vect{\zeta}_i - \vect{\zeta}_j\|$. Similarly, using the smoothness of $\phi(Z)$ and $\nabla \sigma(\vect{\zeta})$, there exists a constant $L_\sigma$ such that $\|\Delta_\sigma\| \leq L_\sigma \|\vect{\zeta}_i - \vect{\zeta}_j\|$. Combining the bounds for $\Delta_\mu$ and $\Delta_\sigma$, we then have that
\begin{align} 
\|\nabla \alpha_{\text{EI}}(\vect{\zeta}_i) - \nabla \alpha_{\text{EI}}(\vect{\zeta}_j)\| \leq (L_\mu + L_\sigma) \|\vect{\zeta}_i - \vect{\zeta}_j\|.
\end{align}
Considering $L = L_\mu + L_\sigma$, the proof is complete. The acquisition function $\alpha_{\text{EI}}(\vect{\zeta}_i)$ is then differentiable, and its gradient satisfies the Lipschitz continuity condition with the Lipschitz constant $L = L_\mu + L_\sigma$.
\end{proof}

\begin{Lemma}
The acquisition function $\alpha_{\text{EI}}$ defined in \eqref{eq:acq_EI} is monotonically decreasing with respect to the exploration term $\alpha_1=\sigma\left(\vect\theta\right)\in[0,1]$ and the exploitation term $\alpha_2=J^\star-\mu\left(\vect\theta\right)$. 
\end{Lemma}
\begin{proof}
Let us rewrite \eqref{eq:acq_EI} as
\begin{align}
  \alpha_{\text{EI}}(\cdot) =-\beta(\cdot),  
\end{align}
where $\beta(\cdot)=\alpha_2\Phi\left(\frac{\alpha_2}{\alpha_1}\right)+\alpha_1\phi\left(\frac{\alpha_2}{\alpha_1}\right)$. To show that $ \alpha_{\text{EI}}(\cdot)$ is monotonically decreasing, one needs to show that $\beta(\cdot)$ is monotonically increasing. To this end, we take the derivative of $\beta(\cdot)$ w.r.t. $\alpha_1$ and $\alpha_2$ as
\begin{align}\label{eq:frac1}
    &\frac{\partial\beta}{\partial\alpha_1}=\\\nonumber
    &-\alpha_2\Phi^'\left(\frac{\alpha_2}{\alpha_1}\right)\left(\frac{\alpha_2}{\alpha_1^2}\right)+\phi\left(\frac{\alpha_2}{\alpha_1}\right)-\alpha_1\phi^'\left(\frac{\alpha_2}{\alpha_1}\right)\left(\frac{\alpha_2}{\alpha_1^2}\right),
\end{align}
\begin{align}\label{eq:frac2}
    \frac{\partial\beta}{\partial\alpha_2}=\Phi\left(\frac{\alpha_2}{\alpha_1}\right)+\left(\frac{\alpha_2}{\alpha_1}\right)\Phi^'\left(\frac{\alpha_2}{\alpha_1}\right)+\phi^'\left(\frac{\alpha_2}{\alpha_1}\right).
\end{align}
Let us consider the following properties of $\Phi$ and $\phi$:
\begin{align}\label{eq:prop_cdf_pdf}
\phi^'\left(x\right)=-x\phi\left(x\right),\quad\Phi^'\left(x\right)=\phi\left(x\right).
\end{align}
By leveraging the properties in \eqref{eq:prop_cdf_pdf}, we can simplify the derivatives \eqref{eq:frac1} and \eqref{eq:frac2} such that the following holds:
\begin{align}
    \frac{\partial\beta}{\partial\alpha_1}=\phi\left(\frac{\alpha_2}{\alpha_1}\right)>0,\quad \frac{\partial\beta}{\partial\alpha_2}=\Phi\left(\frac{\alpha_2}{\alpha_1}\right)>0.
\end{align}
It then follows that $\alpha_{\text{EI}}(\cdot) =-\beta(\cdot)$ is monotonically decreasing since $\beta(\cdot)$ is monotonically increasing.
\end{proof}
Considering the optimization problem \eqref{eq:tot_cost_mabo}, the corresponding augmented Lagrangian is given by
\begin{align}
    &\mathcal{L}(\vect{\zeta}_i, \vect{\bar\zeta}, \vect{\lambda}_i) 
    =\\\nonumber
    &\sum_{i=1}^m \Big( \alpha_{\text{EI}}(\vect\zeta_i) 
    + \vect{\lambda}_i^\top (\vect{\zeta}_i - \vect{\bar\zeta}) 
    + \frac{\rho}{2} \|\vect{\zeta}_i - \vect{\bar\zeta}\|^2 \Big).
\end{align}
By Lemma 2, the augmented Lagrangian $\mathcal{L}$ decreases monotonically under the updates:
\begin{align}
    \mathcal{L}(\vect{\zeta}_i^{k+1}, \vect{\bar\zeta}^{k+1}, \vect{\lambda}_i^{k+1}) 
    \leq \mathcal{L}(\vect{\zeta}_i^k, \vect{\bar\zeta}^k, \vect{\lambda}_i^k),
\end{align}
with equality only at the optimal solution. By Lemma 1, $\nabla \alpha_{\text{EI}}$ then exists so that taking the gradient of the augmented Lagrangian with respect to \(\vect{\zeta}_i\) gives the necessary condition for the optimality of \(\vect{\zeta}_i^{k+1}\) as
\begin{align}
    \nabla \alpha_{\text{EI}}(\vect{\zeta}_i^{k+1}) + \vect{\lambda}_i^k + \rho (\vect{\zeta}_i^{k+1} - \vect{\bar\zeta}^k) = 0.
\end{align}
At convergence, we assume that \(\vect{\zeta}_i^k \to \vect{\zeta}_i^\star\) and \(\vect{\bar\zeta}^k \to \vect{\bar\zeta}^\star\), so the stationarity condition for \(\vect{\zeta}_i^\star\) becomes
\begin{align}
    \nabla \alpha_{\text{EI}}(\vect{\zeta}_i^\star) + \vect{\lambda}_i^\star + \rho (\vect{\zeta}_i^\star - \vect{\bar\zeta}^\star) = 0.
\end{align}

\begin{Remark}
Convergence analysis of the non-convex ADMM is generally a challenging problem due to the absence of Fejér monotonicity \cite{yuan2025admmnonconvexoptimizationminimal}. One approach can be to instead prove that the non-convex $\text{EI}$ function \eqref{eq:acq_EI} is monotonically decreasing and then use Lemmas 1 and 2 to establish two properties of $\alpha_{\text{EI}}$ to facilitate the convergence analysis of the ADMM-based MABO. This is left as a topic for future work.

\end{Remark}


\section{Numerical Examples} \label{sec:5}

To examine the viability of the proposed learning-based distributed MPC scheme, we consider two numerical examples: 1) heterogeneous linear multi-agent systems with simple dynamics, and 2) formation control problem for multiple Wheeled Mobile Robots (WMRs).

\subsection*{Example 1: \textit{Linear Multi-agent System}}

In this example, we consider three linear time-invariant (LTI) systems with different dynamics that must satisfy their local constraint and a coupled equality constraint (which is a desired distance between their first states) in a distributed manner. However, the local state constraints and the coupling constraints would be violated due to the model mismatch and disturbances. Let us consider three agents with the following dynamics
\begin{subequations}
\begin{align}
    &\vect x_1^{k+1}=\begin{bmatrix}
0.9 &0.35 \\ 
 0&1.1 
\end{bmatrix}\vect x_1^{k}+\begin{bmatrix}
0.0813\\0.2 
\end{bmatrix}u_1^{k}+\begin{bmatrix}
e_1^k\\0
\end{bmatrix},\\
&\vect x_2^{k+1}=\begin{bmatrix}
0.91 &0.33 \\ 
 0&0.98 
\end{bmatrix}\vect x_2^{k}+\begin{bmatrix}
0.0611\\0.23 
\end{bmatrix} u_2^{k},\\
&\vect x_3^{k+1}=\begin{bmatrix}
0.88 &0.3 \\ 
 0&1.1 
\end{bmatrix}\vect x_3^{k}+\begin{bmatrix}
0.0837\\0.21 
\end{bmatrix} u_3^{k},
\end{align}
\end{subequations}
and choose an imperfect model for the local MPC schemes as
\begin{align}\label{eq:imp_model}
\vect x^{k+1}=\begin{bmatrix}
1 &0.25 \\ 
 0&1 
\end{bmatrix}\vect x^{k}+\begin{bmatrix}
0.0312\\0.25 
\end{bmatrix} u^{k},
\end{align}
where the disturbance $e_1^{k}$ is random, uncorrelated and uniformly distributed in the interval $\left[-0.1,0\right]$. Let us label the states of the agents as $\vect x_1=\left[x_{1,1},x_{1,2}\right]^\top$, $\vect x_2=\left[x_{2,1},x_{2,2}\right]^\top$ and $\vect x_3=\left[x_{3,1},x_{3,2}\right]^\top$. We then consider the local constraints $0\leq x_{1,1}\leq 0.5$, $0\leq x_{2,1}\leq 2$ and $-2\leq x_{3,1}\leq 0$. The control input constraint $-0.5\leq u_{1,2,3}\leq 0.5$ is considered for all agents, and coupling constraints are defined as relative distances $d_{12}=1.5$, $d_{13}=1.5$ and $d_{23}=3$ (note that $d_{ij}=-d_{ji}$). The local cost function for each agent is defined as $
    L_i(\vect x^k,\vect u_i^k)=
    10\left \|\vect x_i^c\right \|_2^2
    +l_i\left(\vect x_i^k,\vect u_i^k\right)+\vect p^\top\cdot\text{max}\left(0,h_i\left(\vect x_i^k, \vect u_i^k\right)\right)$, where 
$\vect x_i^c:=\left[x_{i,1}-x_{j,1}-d_{i,j}\right]_{j\neq i}\in\mathbb{R}^{m-1}$. The local cost function $l_i$ can be, for instance, a quadratic function, and the penalty vector is set to $\vect p=\left[100,100\right]$.

As observed from Fig. \ref{fig_s1}, the coupling constraints are not satisfied for a DMPC without learning as the desired distances shown in green color cannot be captured by the actual distances shown in cyan color. These coupling constraints can be however satisfied by invoking the proposed learning-based DMPC, shown in red color.
\begin{figure}[htbp!]
\centering
\includegraphics[width=1\linewidth]{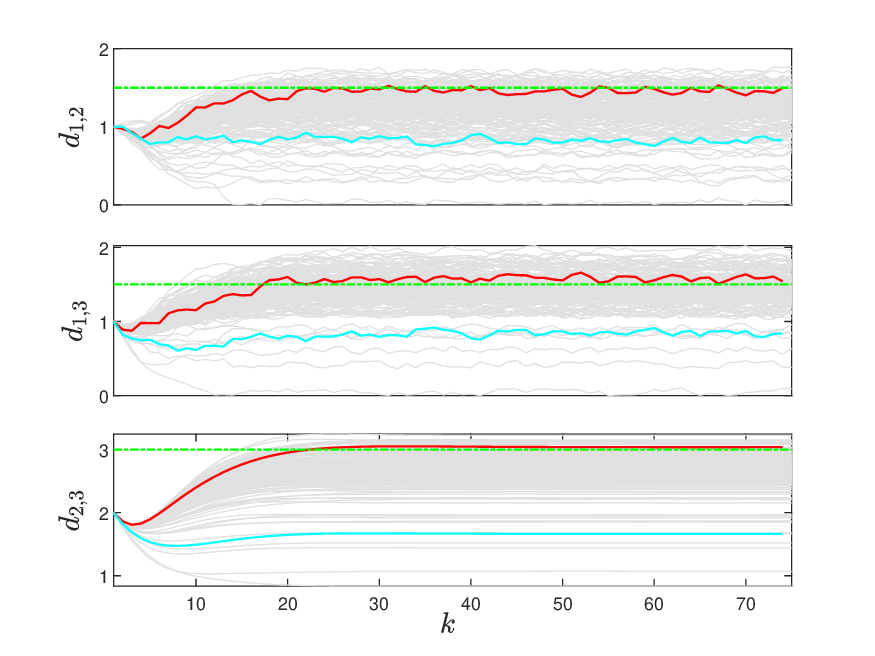}
\caption{The gray lines show the evolution of the actual distances during the learning process. The green lines show the desired distances between the first states of agents. The cyan lines show the evolution of the actual distances when a conventional DMPC without learning is used. The red lines show the evolution of the actual distances using a learned DMPC.} 
\label{fig_s1}
\end{figure}

Fig. \ref{fig_s2} shows the coupling states, which are the first states of the three agents. To make the situation more challenging for the first agent, we simultaneously set a constraint and a desired reference to zero so that the state $x_{1,1}$ must be as close as possible to $0$ while avoiding the yellow unsafe zone. As observed, in the results obtained from the conventional DMPC, the first agent violates the constraint at $\textit{zero}$ due to the disturbance $e_1$ while this constraint violation disappears, and the agent keeps its first state $x_{1,1}$ close to $0$ using the proposed MABO-based DMPC. The three control inputs are shown in Fig. \ref{fig_s3}.
\begin{figure}[htbp!]
\centering
\includegraphics[width=1\linewidth]{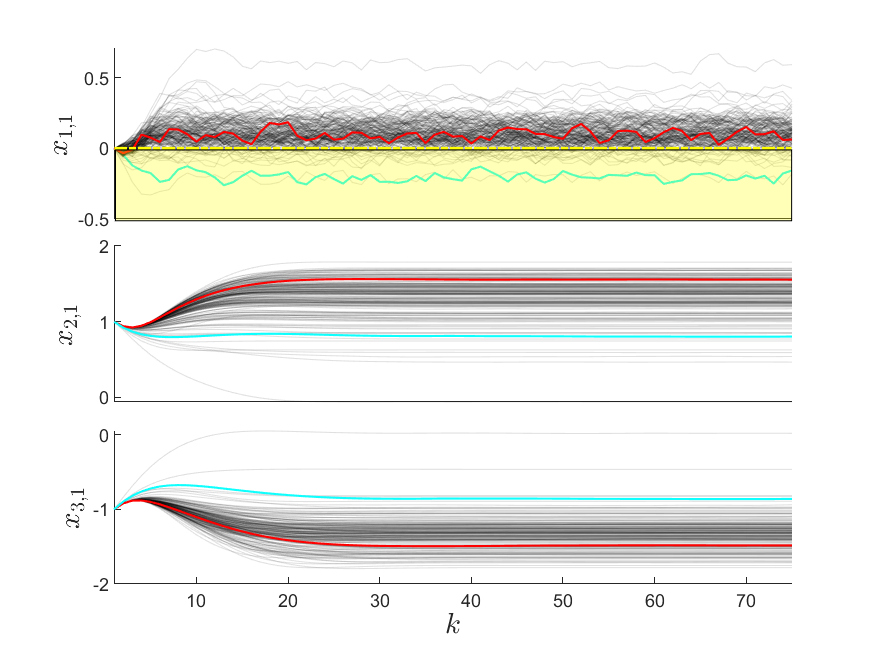}
\caption{The gray lines show the evolution of the first states during the learning process. For the first state of the agent $1$, we consider the point $0$ as reference and constraint simultaneously. The yellow region shows the unsafe zone. The cyan lines show the evolution of the first states when a conventional DMPC without learning is used. The red lines show the evolution of the first states using the proposed learned DMPC scheme.} 
\label{fig_s2}
\end{figure}

\begin{figure}[htbp!]
\centering
\includegraphics[width=1\linewidth]{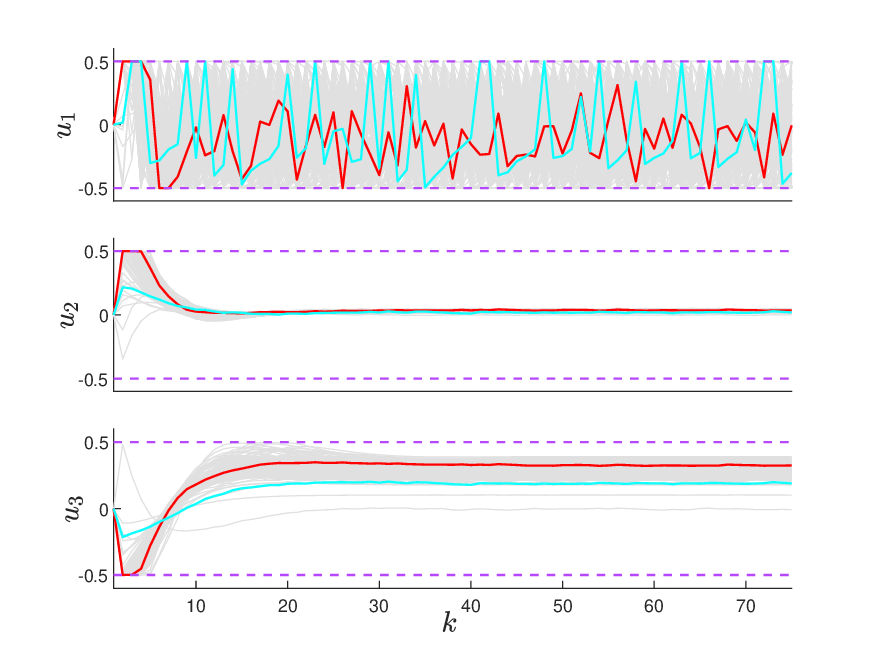}
\caption{The gray lines show the evolution of control signals during the learning process. The cyan lines show the results when a conventional DMPC without learning is used while the red lines show the results obtained from using the learned DMPC.} 
\label{fig_s3}
\end{figure}
As observed in Fig. \ref{fig_s4}, the proposed coordinated learning method using MABO outperforms conventional BO without coordination in achieving the best closed-loop performance for each local MPC scheme. Consequently, we observe that the learned local MPC schemes using BO without coordination cannot perfectly handle the coupling constraints, as shown in Fig. \ref{fig_s5}. Moreover, the BO-DMPC cannot drive the first state of the first agent $x_{1,1}$ to the reference point $0$ accurately, as shown in Fig. \ref{fig_s6}.
\begin{figure}[htbp!]
\centering
\includegraphics[width=1\linewidth]{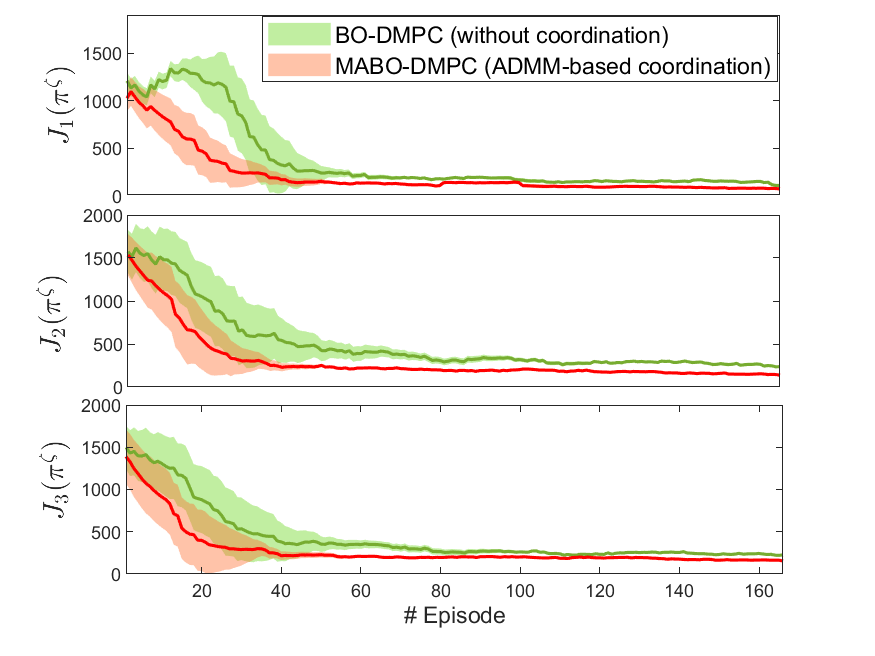}
\caption{The evolution of the local closed-loop performance of agents during the learning process is shown comparing the proposed multi-agent BO (MABO) and conventional BO combined with DMPC.} 
\label{fig_s4}
\end{figure}

\begin{figure}[htbp!]
\centering
\includegraphics[width=1\linewidth]{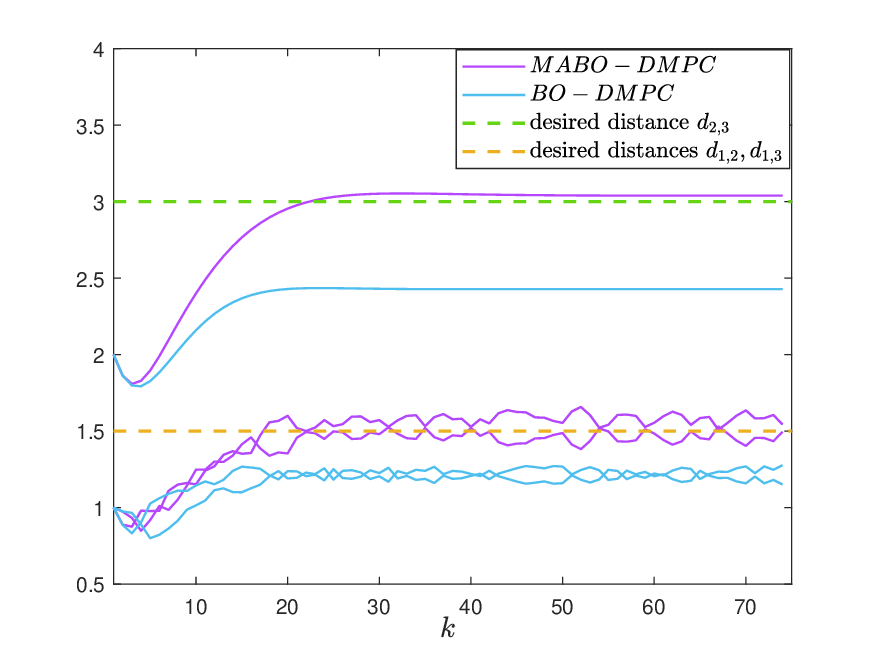}
\caption{The evolution of the coupling constraints is shown here comparing the proposed MABO and conventional BO combined with DMPC.} 
\label{fig_s5}
\end{figure}

\begin{figure}[htbp!]
\centering
\includegraphics[width=1\linewidth]{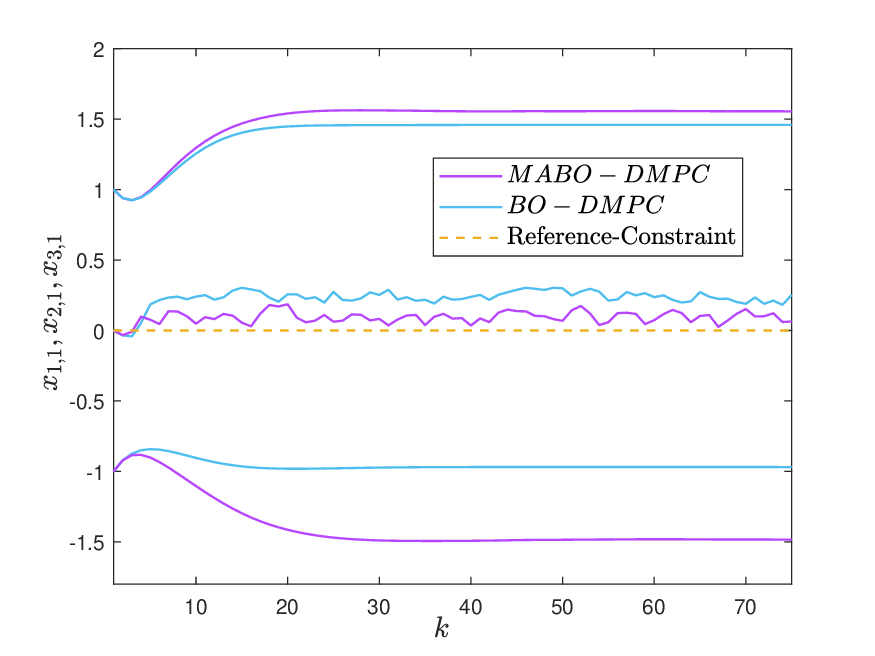}
\caption{The evolution of the first states is shown here comparing the proposed MABO and conventional BO combined with DMPC.} 
\label{fig_s6}
\end{figure}

\subsection*{Example 2: \textit{Formation Control}}

We define the WMR model as
\begin{align}\label{eq:WMR_c}
	&\vect f(\vect x, \vect u)=\begin{bmatrix}
		\cos(\psi) &0 \\ 
		\sin(\psi)& 0\\ 
		0&1 
	\end{bmatrix}\vect u,
\end{align}
where $\vect x=\left[x,y,\psi\right]^\top$ and $\vect u=\left[v,\omega\right]^\top$ are the state and control input vectors, respectively. The position coordinates of the WMR are labeled $x, y$, and $\psi$ is the angle of orientation of the robot. The control inputs $v$ and $\omega$ are the linear and angular velocities, respectively. To discretize the above continuous model, we use a fourth-order Runge-Kutta (RK4) integrator that provides the discretized function $\vect f_d$ of the WMR model as $\vect x(k+1)=\vect f_d\left(\vect x(k),\vect u(k)\right)$. We use two imperfect models of the real system for the second and third local MPC schemes. These two models are affected by the constant uncertainty in their control inputs such that we have the two imperfect models of \eqref{eq:WMR_c} as $\vect f_2\left(\vect x_2,d_2\vect u_2\right)$ and $\vect f_3\left(\vect x_3,d_3\vect u_3\right)$. We select $d_2=0.2,d_3=1.5$ and consider that three WMRs are required to achieve a triangular formation. As observed in Fig. \ref{fig_s7}, the DMPC scheme \eqref{eq:modif_DMPC} without learning cannot achieve the target triangular formation due to the uncertainties in the models that second and third WMRs employ. We then learn the parametric local MPC cost functions using the proposed MABO in order to approach the target formation as shown in Fig. \ref{fig_s8}. As observed in Fig. \ref{fig_s9}, we run the proposed MABO-DMPC for $360$ episodes to achieve the best local closed-loop performance leading to the desired triangular formation.
\begin{figure}[htbp!]
\centering
\includegraphics[width=0.85\linewidth]{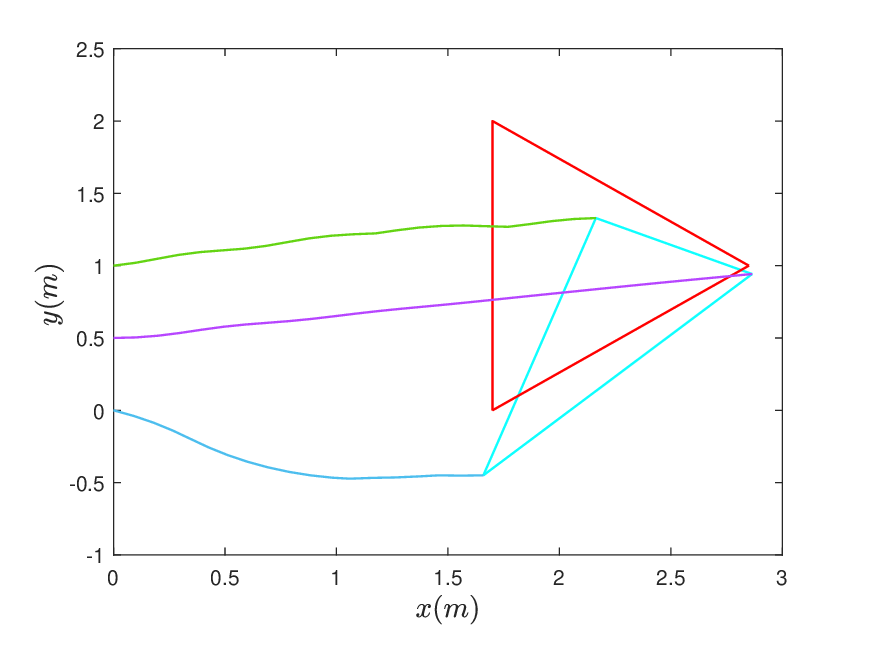}
\caption{Formation control using the proposed DMPC scheme \eqref{eq:modif_DMPC} with wrong MPC models $\vect f_2$ and $\vect f_3$. The red triangle shows the target formation while the cyan triangle is the actual formation. The purple, green, and blue lines show the trajectories of WMR 1, WMR 2, and WMR 3, respectively.} 
\label{fig_s7}
\end{figure}

\begin{figure}[htbp!]
\centering
\includegraphics[width=1\linewidth]{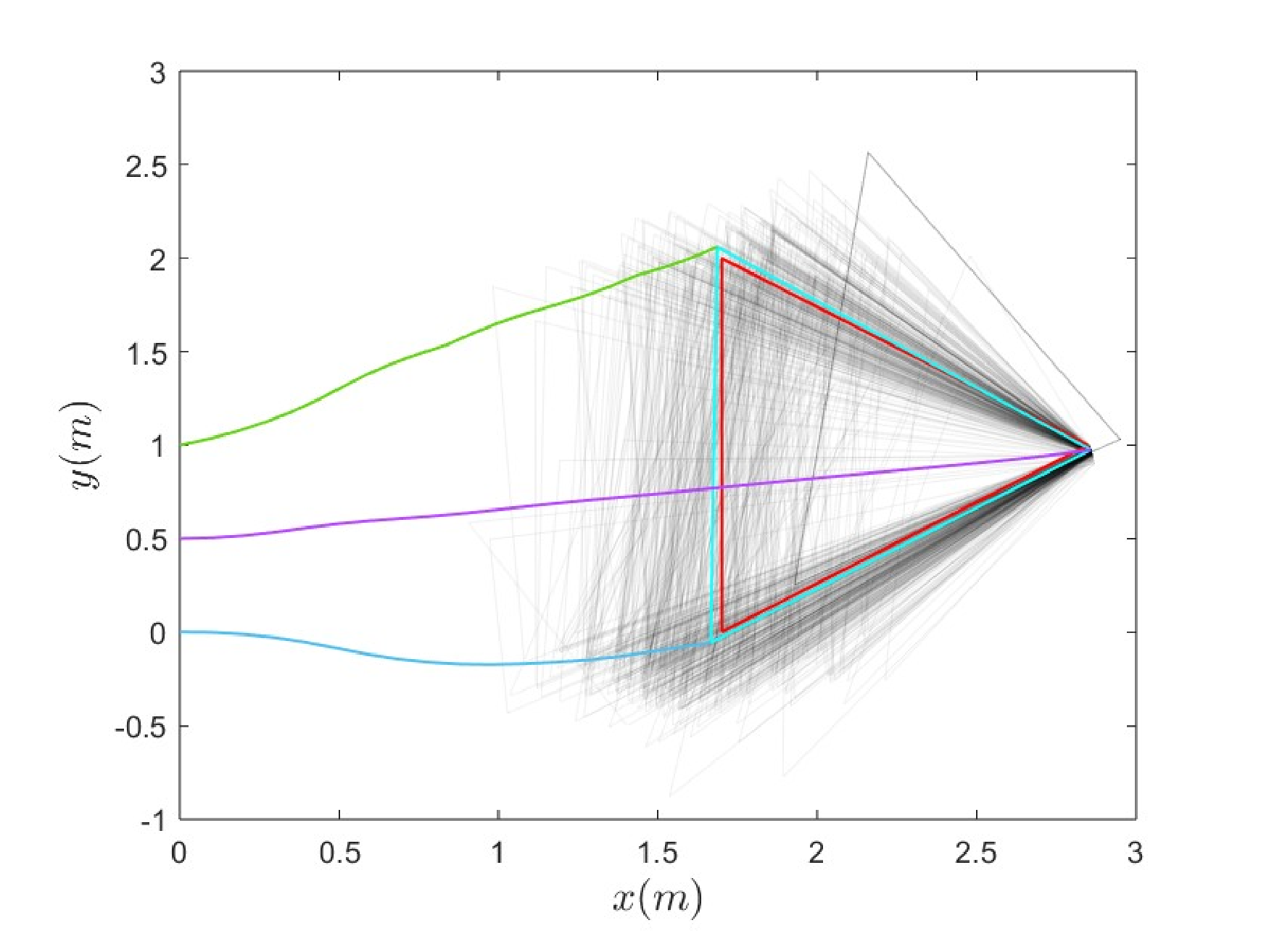}
\caption{Formation control using the proposed learning-based parametric DMPC scheme \eqref{eq:DMPC_par} with wrong MPC models $\vect f_2$ and $\vect f_3$. The red triangle shows the target formation while the cyan triangle is the actual formation achieved after learning. The purple, green, and blue lines show the trajectories of WMR 1, WMR 2, and WMR 3, respectively. The gray lines show the evolution of the triangular formation during the MABO-based learning process.} 
\label{fig_s8}
\end{figure}

\begin{figure}[htbp!]
\centering
\includegraphics[width=0.9\linewidth]{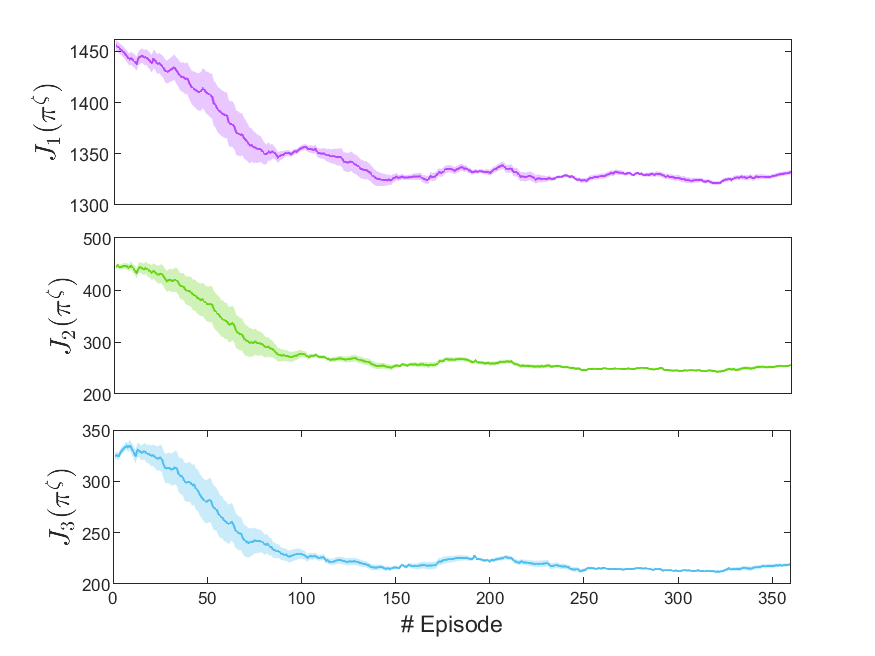}
\caption{The evolution of the local closed-loop performance (cost) of WMRs during the learning process.} 
\label{fig_s9}
\end{figure}


\section{Concluding Remarks} \label{sec:6}

In this paper, we developed a new coordinated learning framework for the design of distributed MPC schemes (DMPC) invoking dual composition assuming that the local MPC models are unable to capture the real multi-agent system accurately. The proposed method leveraged the machinery of multi-agent Bayesian optimization leading to a coordinated framework for learning the parameterized DMPC aiming at improving the local closed-loop performance. We observed significant improvements in allowing each agent to refine its performance in a coordinated manner so that the global task the multi-agent system is assigned to is satisfied.

\bibliographystyle{IEEEtran}

\bibliography{IEEEabrv,ref}

\end{document}